\documentclass[a4paper]{amsart}

\title[A Linear Category of Polynomial Diagrams]
      {A Linear Category of Polynomial Diagrams}

\author{Pierre Hyvernat}
\address{Laboratoire de Math\'ematiques\\
  CNRS UMR 5126 -- Universit\'e de Savoie\\
  73376 Le Bourget-du-Lac Cedex\\
  France}
\email{\href{mailto:pierre.hyvernat@univ-savoie.fr}{pierre.hyvernat@univ-savoie.fr}}
\urladdr{\url{http://lama.univ-savoie.fr/~hyvernat/}}


\date{September 2012}

\usepackage[colorlinks=true,citecolor=blue,linkcolor=blue]{hyperref}
\usepackage{amssymb,amsmath}

\newif\ifUSEEXTERNALGRAPHICS
\USEEXTERNALGRAPHICStrue

\ifUSEEXTERNALGRAPHICS
  \usepackage{tikzexternal}
  \tikzsetexternalprefix{figures-poly_diagrams/}

\else
  \usepackage{tikz}
  \usetikzlibrary{matrix,arrows,calc}

  \usetikzlibrary{external}
  \tikzset{morphism/.style={->,font=\scriptsize}}
  \tikzset{injection/.style={right hook->,font=\scriptsize}}
  \tikzset{twocell/.style = {double equal sign distance,double,-implies,shorten <= .3cm,shorten >=.3cm,font=\scriptsize,draw}}
  \tikzset{crossover/.style={preaction={solid,draw=white,-,line width=6pt}}}
  \tikzset{over/.style={fill=white,inner sep=1.5pt}}

  \newcommand{\pullback}[5][.7cm]{%
    \node[coordinate] (a) at (#2) {} ; %
      \node[coordinate] (b) at (#3) {} ; %
      \node[coordinate] (c) at (#4) {} ; %
      \path let
      \p1= ($(a) - (b)$) , %
      \p2= ($(c) - (b)$), %
      \n1={veclen(\x1,\y1)}, %
      \n2={veclen(\x2,\y2)}, %
      \n3={min(#1, .4 * min(\n1,\n2))},%
      \p3=($(\x1/\n1,\y1/\n1)$),%
      \p4=($(\x2/\n2,\y2/\n2)$),%
      \p5=($(\n3 * \x3, \n3 * \y3)$),%
      \p6=($(\n3 * \x4, \n3 * \y4)$) in%
      (b) ++ (\p5)  node[coordinate] (x) {} 
      (b) ++ (\p6)  node[coordinate] (y) {} 
      ;
    \node[coordinate] (c) at (barycentric cs:x=1,y=1,b=-1) {} ; %
      \path (x) -- node[pos=.02] (xup){} (c) ; %
      \path (y) -- node[pos=.02] (yup){} (c) ; %
      \path[#5] (xup) -- (c) -- (yup) ; %
  }
\fi

\newtheorem{prop}{Proposition}[section]
\newtheorem{lem}{Lemma}[section]
\newtheorem{defi}{Definition}[section]

\newcommand\inter[2]{[\![#2]\!]\,=\,#1}
\newcommand\Rule[3]{\frac{\Big.\quad{#1}\quad}{\Big.#2}\quad\hbox{\scriptsize #3}}

\newcommand\C{\mathbb{C}}
\newcommand\Set{\mathsf{Set}}
\newcommand\Rel{\mathsf{Rel}}
\newcommand\Span{\mathsf{Span}}
\newcommand\Poly{\mathsf{Poly}}
\newcommand\PolyFun{\mathsf{PolyFun}}
\newcommand\PE{\mathsf{PSim}}
\newcommand\PEFun{\mathsf{FSim}}
\newcommand\Hom{\mathrm{Hom}}
\newcommand\Sl[2]{{#1}/_{\!#2}}       

\newcommand\iso{\cong}
\DeclareMathOperator{\id}{id}

\newcommand\Zero{\mathbf{0}}
\newcommand\One{\mathbf{1}}
\newcommand\Plus{\oplus}
\newcommand\Tensor{\otimes}
\newcommand\Linear{\multimap}
\newcommand\Bang{\textbf{!}}
\newcommand\exclamation{\hbox{\rm!}}

\newcommand\eqdef{\stackrel{\smash{\text{\sf def}}}{=}}
\newcommand\BLANK{{\texttt{\char"5F}}}
\newcommand\at{@}
\newcommand\word[1]{\mathfrak{#1}}
\newcommand\dd{\partial}

\newcommand\inl{\mathop{\mathrm{inl}}}

\newcommand\Mf{\ensuremath{\mathcal{M}_{\!f}}}

\newcommand\Sem[1]{#1}          

\begin{document}

\begin{abstract} 

  We present a categorical model for intuitionistic linear logic where
  {objects} are polynomial diagrams and morphisms are \emph{simulation
  diagrams}. The multiplicative structure (tensor product and its adjoint) can
  be defined in any locally cartesian closed category, whereas the additive
  (product and coproduct) and exponential ($\Tensor$-comonoid comonad)
  structures require additional properties and are only developed in the
  category~$\Set$, where the objects and morphisms have natural
  interpretations in terms of games, simulation and strategies.

\end{abstract} 

\maketitle

\section*{Introduction}  

Categories of games abound in the literature of denotational semantics of
linear logic. We present a category based on a notion of game that differs
from traditional games semantics. Objects are a kind of two players game and
following a well established tradition \cite{Joyal}, morphisms from~$G_1$
to~$G_2$ amount to strategies (for the first player) in a game
called~$G_1\Linear G_2$. Composition of strategies is a ``relational
composition'', or more precisely, ``span composition''. In particular,
interaction is irrelevant to the definition of composition. The theory is
developed in the abstract setting of locally cartesian closed categories and
polynomial diagrams \cite{polyMonads,notesJoachim}.

The paper is organized as follows: after some preliminaries
(section~\ref{section:prelim}), we construct a symmetric monoidal closed
category around polynomial diagrams over any locally cartesian closed category
(section~\ref{section:PE}). We then restrict to the case where the base
category is the category of sets and functions and extend the SMCC structure
to a denotational model for full intuitionistic linear logic by adding a
biproduct (cartesian and cocartesian structure) and constructing an
exponential comonad for the free commutative~$\Tensor$-comonoid
(section~\ref{section:ISinSet}). The proof that the exponential structure is
indeed the free commutative~$\Tensor$-comonoid involves a lot of tedious
checking and, for simplicity's sake, is only spelt out in a simpler category.


\subsection*{Related Works} 

Part of this work is implicitly present in~\cite{PhD,giovanni} where
polynomial endofunctors are called ``interaction systems'' and everything was
done with dependant type theory. The focus was on representing formal
topological spaces with dependent types. Categorically speaking, it amount to
the following: it is possible to show that the free monad construction on
polynomial endofunctors described in the second part of~\cite{polyMonads}
gives a monad on the category of polynomial endofunctors and simulations
(Section~\ref{sub:simulations}). The category introduced in~\cite{giovanni} is
``simply'' the Kleisli category of the this monad.

The same kind of polynomial functors are also used by Altenkirch and Morris
(together with Ghani, Hancock and McBride) in~\cite{indexed-containers} under
the name ``indexed containers'' in order to give a semantics to a large family
of indexed, strictly positive datatypes. The unindexed version was developed
earlier in~\cite{cont}.


\subsection*{Games} 
\label{sub:games}

The games we consider are non-initialized, state-based, 2-players, alternating
games. More precisely, a game is given by the following data:
\begin{itemize}
  \item a set~$I$ of \emph{states},
  \item for each state~$i\in I$, a set of \emph{moves}~$A(i)$ for Alfred,
  \item for each move~$a\in A(i)$, another set of \emph{counter moves}~$D(a)$
    for Dominic,
  \item a function~$n$ going from counter moves to states, giving the new
    state after each possible choice of counter move from Dominic. We usually
    write~$i[a/d]$ instead of~$n(d)$ whenever~$a\in A(i)$ and~$d\in D(a)$.
\end{itemize}
As is customary in categories, we represent a family of sets indexed by~$I$ by
an object of the slice category~$\Sl\Set I$. A game is thus simply given by
a diagram of the form
\[
  \begin{tikzpicture}[baseline=(m-1-1.base)]
    \matrix (m) [matrix of math nodes, column sep=3em]
      { I & D & A & I \\ };
    \path[morphism]
      (m-1-2) edge node[above]{$n$} (m-1-1)
      (m-1-2) edge node[above]{$d$} (m-1-3)
      (m-1-3) edge node[above]{$a$} (m-1-4);
  \end{tikzpicture}
  \ .
\]
Other names for the players would be ``Player'' and ``Opponent'' (games
semantics), ``Angel'' and ``Demon'' (process calculus), ``Alice'' and ``Bob''
(cryptography) etc.

\medbreak
Those games are slightly asymmetrical in that there is no actual state
between~$A$-moves and~$D$-moves. In particular, it is not obvious what form
the familiar~$A$/$D$ duality should take.
Many of the usual board games such as chess or go are more symmetric: given a
state, both players could make a move. Such games are more appropriately
represented by two spans over the same set:
\[
  \begin{tikzpicture}[baseline=(m-1-1.base)]
    \matrix (m) [matrix of math nodes, column sep=3em]
      { I & D & I & & I & A & I \\ };
    zc\path[morphism]
      (m-1-2) edge node[above]{$d$} (m-1-1)
      (m-1-2) edge node[above]{$n_d$} (m-1-3)
      (m-1-6) edge node[above]{$a$} (m-1-5)
      (m-1-6) edge node[above]{$n_a$} (m-1-7);
  \end{tikzpicture}
  \ .
\]
\begin{itemize}
  \item $I$ represents the set of states,
  \item for each state~$i\in I$, the fiber~$A(i)$ gives the available~$A$-moves and the
    function~$n_A$ gives the new state after such a move,
  \item for each state~$i\in I$, the fiber~$D(i)$ gives the available~$D$-moves and the
    function~$n_D$ gives the new state after such a move.
\end{itemize}
We get an asymmetric alternating version as above with the plain arrows
from the following construction:
\[
  \begin{tikzpicture}[baseline=(m-2-1.base)]
    \matrix (m) [matrix of math nodes, column sep={4em,between origins}, row
    sep={4em,between origins}, text height=1.5ex, text depth=.25ex]
      {   &&   & D \times_{\scriptscriptstyle\!I} A &   & A &   \\
        I && D &                                    & I &   & I \\   };
    \path[morphism,bend right=15]
      (m-1-4) edge (m-2-1);
    \path[morphism]
      (m-1-4) edge (m-1-6)
      (m-1-6) edge node[above right]{$a$} (m-2-7);
    \path[morphism,densely dashed]
      (m-1-4) edge (m-2-3)
      (m-2-3) edge node[below]{$n_d$} (m-2-1)
      (m-2-3) edge node[below]{$d$} (m-2-5)
      (m-1-6) edge node[below right]{$n_a$} (m-2-5);
    \pullback{m-2-3}{m-1-4}{m-1-6}{draw,-}
  \end{tikzpicture}
  \ .
\]
This kind of games could prove an interesting starting point for representing
Conway games \cite{Joyal}.
For those games, the opposite (or negation for games semantics) amounts to
interchanging the players by considering
\[
  \begin{tikzpicture}[baseline=(m-2-1.base)]
    \matrix (m) [matrix of math nodes, column sep={4em,between origins}, row
    sep={4em,between origins}, text height=1.5ex, text depth=.25ex]
      {   &&   & D \times_{\scriptscriptstyle\!I} A &   & A &   \\
        I && D &                                    & I &   & I \\   };
    \path[morphism,bend right=15]
      (m-1-4) edge (m-2-1);
    \path[morphism]
      (m-1-4) edge (m-1-6)
      (m-1-6) edge node[above right]{$d$} (m-2-7);
    \path[morphism,densely dashed]
      (m-2-3) edge node[below]{$n_a$} (m-2-1)
      (m-1-4) edge (m-2-3)
      (m-2-3) edge node[below]{$a$} (m-2-5)
      (m-1-6) edge node[below right]{$n_d$} (m-2-5);
    \pullback{m-2-3}{m-1-4}{m-1-6}{draw,-}
  \end{tikzpicture}
\]
which is \emph{not} the dual of section~\ref{sec:dual}.

\smallbreak
Another possibility suggested by Martin Hyland would be to consider games where
the two players play simultaneously. Those games are represented by a
slice~$a:A\to I$ for~$A$-moves and a slice~$d:D\to I$ for~$D$-moves, together with a
function~$n:A\times_{\scriptscriptstyle\!I}D \to I$ to get the next state.
As far as our games are concerned, it amounts to considering diagrams of the
form
\[
  \label{rk:HylandGames}
  \begin{tikzpicture}[baseline=(m-1-1.base)]
    \matrix (m) [matrix of math nodes, column sep=4em, text height=1.5ex, text depth=.25ex]
      {   I & A \times_{\scriptscriptstyle\!I} D & D  & I \\   };
    \path[morphism]
      (m-1-2) edge node[above]{$n$} (m-1-1)
      (m-1-2) edge node[above]{$\Delta_{d}(a)$} (m-1-3)
      (m-1-3) edge node[above]{$a$} (m-1-4);
  \end{tikzpicture}
  \ .
\]
In such games, the sets of counter moves~$D(a)$ depend only on~$i\in
I$ and not on the actual~$a\in A(i)$. Here again, there is a natural notion of
duality, different from the one from section~\ref{sec:dual}.


\subsection*{Strategies and Simulations}  
\label{sub:strategies}

For a game~$I\leftarrow D \rightarrow A \rightarrow I$ as above, a
``non-losing strategy for Alfred'' consists of:
\begin{itemize}
  \item a subset of ``good'' states~$H\subset I$,
  \item a function~$\alpha$ choosing an~$A$-move for each~$i\in H$,
  \item such that whenever~$i\in H$ and~$d\in D\big(\alpha(i)\big)$, we
    have~$i[\alpha(i)/d]\in H$.
\end{itemize}
This means that, as long as the games start in~$H$, Alfred always has a move
to play. Each game will either go on infinitely or stop when Dominic has no
counter-move available. Categorically speaking, such a strategy is described
by a diagram
\[
  \begin{tikzpicture}[baseline=(m-2-1.base)]
    \matrix (m) [matrix of math nodes, column sep={4em,between origins}, row
    sep={4em,between origins},text height=1.5ex, text depth=.25ex]
      { H & H{\cdotp}A & H &  \\
        I & D          & A & I\\};
    \path[morphism]
      (m-2-2) edge node[below]{$n$} (m-2-1)
      (m-2-2) edge node[below]{$d$} (m-2-3)
      (m-2-3) edge node[below]{$a$} (m-2-4);
    \path[injection,densely dashed]
      (m-1-3) edge node[above right]{$h$} (m-2-4)
      (m-1-1) edge node[left]{$h$} (m-2-1);
    \path[morphism,densely dashed]
      (m-1-2) edge node[above]{$\gamma$} (m-1-1)
      (m-1-2) edge (m-1-3)
      (m-1-3) edge node[left]{$\alpha$} (m-2-3)
      (m-1-2) edge (m-2-2);
    \pullback[.6cm]{m-2-2}{m-1-2}{m-1-3}{draw,-};
  \end{tikzpicture}
  \ .
\]
Dually, a ``non-losing strategy for Dominic'' is given by a diagram
\[
  \begin{tikzpicture}[baseline=(m-2-1.base)]
    \matrix (m) [matrix of math nodes, column sep={4em,between origins}, row
    sep={4em,between origins},text height=1.5ex, text depth=.25ex]
      { I & D  & A & I \\
        H & H' & H'& H \\};
    \path[morphism]
      (m-1-2) edge node[above]{$n$} (m-1-1)
      (m-1-2) edge node[above]{$d$} (m-1-3)
      (m-1-3) edge node[above]{$a$} (m-1-4);
    \path[injection,densely dashed]
      (m-2-1) edge node[left]{$h$} (m-1-1)
      (m-2-4) edge node[right]{$h$} (m-1-4);

    \path[morphism,densely dashed]
      (m-2-2) edge node[below]{$\gamma$} (m-2-1)
      (m-2-3) edge node[below]{$\alpha$} (m-2-4)
      (m-2-3) edge (m-1-3)
      (m-2-2) edge (m-1-2);
    \draw[double equal sign distance,double]
      (m-2-2) -- (m-2-3);
    \pullback[.6cm]{m-1-3}{m-2-3}{m-2-4}{draw,-};
  \end{tikzpicture}
  \ .
\]

The \emph{simulations} we will consider generalize both kind of strategies and
satisfy a property similar to what appears in the theory of labeled transition
systems, but with an additional layer of quantifiers to account for the
counter moves. To make a relation~$R\subseteq I_1\times I_2$ into a simulation
between games~$I_1\leftarrow D_1\rightarrow A_1\rightarrow I_1$
and~$I_2\leftarrow D_2\rightarrow A_2\rightarrow I_2$ as above, we need two
functions~$\alpha$ and~$\beta$ satisfying:
\begin{itemize}
  \item whenever $(i_1,i_2)\in R$ and $a_1\in A_1(i_1)$, there is a move
    $a_2\eqdef \alpha(a_1)\in A_2(i_2)$ which simulates~$a_1$ in the following
    sense;
  \item whenever $d_2\in D_2(a_2)$ is a response to~$a_2$, there is a response~$d_1\eqdef \beta(d_2)\in
  D_1(a_1)$, such that~$\big(i_1[a_1/d_1]\,,\,i_2[a_2/d_2]\big)\in R$.
\end{itemize}
Strategies for Alfred or Dominic are obtained by instantiating~$I_1\leftarrow
D_1\rightarrow A_1\rightarrow I_1$ or~$I_2\leftarrow D_2\rightarrow
A_2\rightarrow I_2$ to the trivial game~$\One\leftarrow \One\rightarrow
\One\rightarrow \One$, where~$\One=\{\star\}$ is the terminal object of~$\Set$.
The symmetric monoidal closed structure will in particular imply that a
simulation from~$G_1$ to~$G_2$ is simply a strategy (for Alfred) for a game
called~$G_1 \Linear G_2$.

\smallbreak
Note that the actual definition of simulation will be slightly more general in
that it considers arbitrary spans instead of relations (monic spans). For
strategies, considering some~$H\to I$ instead of a subset makes it possible
for the function~$\alpha$ to choose move depending on more than just a state.



\section{Preliminaries, Polynomials and Polynomial Functors}  
\label{section:prelim}

\subsection{Locally Cartesian Closed Categories} 

Some basic knowledge about locally cartesian closed categories and their internal
language (extensional dependent type theory) is assumed throughout the paper.
Here is a review of the notions (and notations) needed in the rest of the
paper.
For a category~$\C$ with finite limits, we write~``$\One$'' for its terminal
object and~``$A\times B$'' for the cartesian product of~$A$ and~$B$.
The ``pairing'' of~$f:C\to A$ and~$g:C\to B$ is written~$\langle f,g\rangle :
C\to A\times B$.

\smallbreak
If~$f:A\to B$ is a morphism, it induces a pullback functor~$\Delta_f$ from
slices over~$B$ to slices over~$A$. This functor has a left adjoint~$\Sigma_f$
which is simply ``pre-composition by~$f$''. When all the~$\Delta_f$s also
have a right adjoint, we say that~$\C$ is \emph{locally cartesian closed}.
The right adjoint is written~$\Pi_f$.  In all the sequel, $\C$ stands for a
category which is (at least) locally cartesian closed.  We thus have
\[
  \Sigma_f \quad\dashv\quad \Delta_f \quad\dashv\quad \Pi_f \quad .
\]
Additional requirements will be explicitly stated.

Besides the isomorphisms coming from the adjunctions, slices enjoy two
fundamental properties:
\begin{itemize}
  \item \emph{the Beck-Chevalley isomorphisms:}
    \[
      \Pi_g \, \Delta_l  \quad\iso\quad  \Delta_k \, \Pi_f
      \qquad\hbox{and}\qquad
      \Sigma_g \, \Delta_l  \quad\iso\quad  \Delta_k \, \Sigma_f
    \]
    whenever
    \[
      \begin{tikzpicture}[baseline=(m-1-1.base)]
        \matrix (m) [matrix of math nodes, row sep={4em,between origins},
        column sep={4em,between origins}]
          { \cdot & \cdot \\ \cdot & \cdot \\ };
        \path[morphism]
          (m-1-1) edge node[auto]{$g$} (m-1-2)
          (m-1-1) edge node[left]{$l$} (m-2-1)
          (m-2-1) edge node[below]{$f$} (m-2-2)
          (m-1-2) edge node[auto]{$k$} (m-2-2);
        \pullback[.5cm]{m-2-1}{m-1-1}{m-1-2}{draw,-};
      \end{tikzpicture}
    \]
    is a pullback,

  \item \emph{distributivity:} when $b:C\to B$ and $a:B\to
    A$, we have a commuting diagram
    \begin{equation}\label{diag:distr}
      \begin{tikzpicture}[baseline=(m-2-1.base)]
        \matrix (m) [matrix of math nodes, row sep={2.5em,between origins},
        column sep={2.5em,between origins}]
          {   & \cdot && \cdot \\
            C &       &&       \\
              &  B    && A     \\};
        \path[morphism]
          (m-1-2) edge node[auto]{$a'$} (m-1-4)
          (m-1-2) edge (m-3-2)
          (m-1-2) edge node[over]{$\epsilon$} (m-2-1)
          (m-2-1) edge node[below left]{$b$} (m-3-2)
          (m-3-2) edge node[below]{$a$} (m-3-4)
          (m-1-2) edge node[left]{$u'$} (m-3-2)
          (m-1-4) edge node[auto]{$u \eqdef \Pi_a(b)$} (m-3-4);
        \pullback[.5cm]{m-3-2}{m-1-2}{m-1-4}{draw,-};
      \end{tikzpicture}
    \end{equation}
    where~$\epsilon$ is the co-unit of~$\Delta_a\dashv \Pi_a$. For such a
    diagram, we have
    \[
      \Pi_a \, \Sigma_b \quad\iso\quad \Sigma_{u} \, \Pi_{a'} \, \Delta_\epsilon
      \ .
    \]

\end{itemize}

\medbreak
Moreover, any slice category~$\Sl\C I$ is canonically enriched over~$\C$ by
putting\label{sub:enrichment}
\[
  \Hom(x,y) \quad \eqdef \quad \Pi_I\Pi_x\Delta_x(y)
\]
whenever~$x,y\in\Sl\C I$. (Here,~$I$ also stands for the unique map from~$I$
to~$\One$.) Any~$\Sl\C I$ is also canonically tensored over~$\C$ by using the
left adjoint of~$\Hom(x,\BLANK)$:
\[
  A \odot x \quad \eqdef \quad \Sigma_x\Delta_x\Delta_I(A)
\]
for any~$x\in\Sl\C I$ and object~$A$.


\subsection{Dependent Type Theory} 

In~\cite{Seely}, Seely showed how an extensional version of Martin L\"of's
theory of dependent types \cite{ML84} could be regarded as the internal
language for locally cartesian closed categories. A little later, Hofmann
showed in~\cite{Hofmann} that Seely's interpretation works only ``up-to
canonical isomorphisms'' and proposed a solution.
Some of the proofs in this paper rely on the use of this internal language and
we use Seely's original interpretation. Strictly speaking, some of those
morphisms constructed with type theory should be composed with canonical
``substitution'' isomorphisms.

A type~$A$ in context~$\Gamma$, written~$\Gamma\vdash A$ is interpreted as a
morphism~$a:\Gamma_{\!A}\to\Gamma$, that is as an object in the slice over
(the interpretation of)~$\Gamma$. Then, a term of type~$A$ in context~$\Gamma$,
written~$\Gamma\vdash t:A$ is interpreted as a
morphism~$u:\Gamma\to\Gamma_{\!A}$ such that~$a u=1$, i.e., a section of (the
interpretation of) its type.  When~$A$ is a type in context~$\Gamma$, we
usually write~$A(\gamma)$ to emphasize the dependency on the context and we
silently omit irrelevant parameters.
If we write~$\inter{a}{\Gamma\vdash A}$ to mean that the interpretation
of type~$\Gamma\vdash A$ is~$a$, the main points of the Seely semantics are:
\[
  \Rule
  {\inter{a}{\Gamma \vdash A} \qquad \inter{b}{\Gamma,x:A\vdash B(x)}}
  {\inter{\Pi_a(b)}{\Gamma \vdash \prod_{x:A} B(x)}}
  {product}
  \ ,
\]
\[
  \Rule
  {\inter{a}{\Gamma \vdash A} \qquad \inter{b}{\Gamma,x:A\vdash B(x)}}
  {\inter{\Sigma_a(b)}{\Gamma \vdash \sum_{x:A} B(x)}}
  {sum}
  \ ,
\]
\[
  \Rule
  {\inter{f}{\Gamma \vdash \vec u:\Delta} \qquad \inter{u}{\Delta\vdash
  A(\vec x)}}
  {\inter{\Delta_f(u)}{\Gamma \vdash A(\vec u)}}
  {substitution}
  \ .
\]
Of particular importance is the distributivity condition~$(\ref{diag:distr})$
whose type theoretic version is an intensional version of the axiom of choice:
\begin{equation}\label{diag:AC}
  \Gamma \ \vdash \ \prod_{x:A} \ \sum_{y:B(x)} U(x,y)
  \qquad\iso\qquad
  \Gamma \ \vdash \ \sum_{f:\prod_{x:A} B(x)} \ \prod_{x:A} \ U\big(x,f(x)\big)
  \ .
\end{equation}


\subsection{Polynomials and Polynomial Functors} 

We now recall some definitions and results from~\cite{polyMonads} and refer to
the original article for historical notes, details about proofs and additional
comments.

\begin{defi}\label{def:polynomial}
  If $I$ and $J$ are objects of~$\C$, a (generalized) \emph{polynomial} from~$I$ to~$J$ is a
  diagram~$P$ in~$\C$ of the shape
  \[
    P \quad=\quad
    \begin{tikzpicture}[baseline=(m-1-1.base)]
      \matrix (m) [matrix of math nodes, column sep=3em]
      { I & D & A & J \\ };
      \path[morphism]
      (m-1-2) edge node[above]{$n$} (m-1-1)
      (m-1-2) edge node[above]{$d$} (m-1-3)
      (m-1-3) edge node[above]{$a$} (m-1-4);
    \end{tikzpicture}
    \ .
  \]
  We write~$\Poly_\C[I,J]$ for the collection of such polynomials from~$I$ to~$J$.
\end{defi}

\begin{defi}
  For each~$P\in\Poly_\C[I,J]$ as in Definition~\ref{def:polynomial}, there is
  an associated functor from~$\Sl\C I$ to~$\Sl\C J$ called the
  \emph{extension} of~$P$. It is also denoted by~$\Sem{P}$ and is defined as
  the following composition
  \[
    \Sem{P} \quad=\quad
    \begin{tikzpicture}[baseline=(m-1-1.base)]
      \matrix (m) [matrix of math nodes, column sep=3em]
      { \Sl\C I & \Sl\C D & \Sl\C A & \Sl\C J \\ };
      \path[morphism]
      (m-1-1) edge node[above]{$\Delta_n$} (m-1-2)
      (m-1-2) edge node[above]{$\Pi_d$} (m-1-3)
      (m-1-3) edge node[above]{$\Sigma_a$} (m-1-4);
    \end{tikzpicture}
    \ .
  \]
  Any functor isomorphic to the extension of a polynomial is called a
  \emph{polynomial functor}. We write~$\PolyFun_\C[I,J]$ for the collection of
  polynomial functors.
\end{defi}

The identity functor from~$\Sl\C I$ to itself is trivially the extension of the
polynomial
\[
  \begin{tikzpicture}[baseline=(m-1-1.base)]
    \matrix (m) [matrix of math nodes, column sep=3em]
    { I & I & I & I \\ };
    \path[morphism]
    (m-1-2) edge node[above]{$1$} (m-1-1)
    (m-1-2) edge node[above]{$1$} (m-1-3)
    (m-1-3) edge node[above]{$1$} (m-1-4);
  \end{tikzpicture}
\]
and it can be shown that polynomial functors compose, see~\cite{polyMonads}
for example. We thus obtain a category~$\PolyFun_\C$ where objects are slice
categories and morphisms are polynomial functors. We also obtain a
\emph{bicategory}~$\Poly_\C$ of polynomials: composition of polynomials is
associative only up-to canonical isomorphisms.

\begin{prop} \label{prop:PFConnectedLimits}
  Polynomial functors commute with connected limits.
\end{prop}
The simplest example of functor which is not polynomial is the \emph{finite
multiset functor}, seen as a functor from~$\Sl\Set\One \iso \Set$ to itself:
this functor doesn't commute with connected limits.
When~$\C$ is~$\Set$, the converse of proposition~\ref{prop:PFConnectedLimits}
also holds, giving a more ``extensional'' characterization of polynomial
functors:
\begin{prop}
  A functor $P:\Sl\Set I \to \Sl\Set J$ is polynomial iff it commutes with all
  connected limits.
\end{prop}
There are several other characterizations of polynomial functors on~$\Set$,
all nicely summarized in~\cite{polyMonads}.

%
%


\section{Polynomials and Simulations: SMCC Structure}  
\label{section:PE}

We will now construct a category where polynomials play the r\^ole of objects.
More precisely, we will consider ``endo-polynomials'', i.e., diagrams of the form
\[
  \begin{tikzpicture}[baseline=(m-1-1.base)]
    \matrix (m) [matrix of math nodes, column sep=3em]
    { I & D & A & I \\ };
    \path[morphism]
    (m-1-2) edge node[above]{$n$} (m-1-1)
    (m-1-2) edge node[above]{$d$} (m-1-3)
    (m-1-3) edge node[above]{$a$} (m-1-4);
  \end{tikzpicture}
  \ .
\]
We simply call such a diagram a \emph{polynomial over~$I$} and think of
them as games, as described on page~\pageref{sub:games}.

\subsection{Simulations} 
\label{sub:simulations}

The morphisms between two such polynomials over~$I$ and~$J$ will be spans
between~$I$ and~$J$, with some additional structure.
\begin{defi}\label{def:simulationAsNatTransformation}
  If $P_1$ and $P_2$ are two polynomial functors over~$I_1$ and~$I_2$
  respectively, a \emph{simulation from~$P_1$
  to~$P_2$} is a diagram
  \[
    \begin{tikzpicture}[baseline=(m-3-2)]
      \matrix (m) [matrix of math nodes, column sep={3.5em,between origins}, row
      sep={3.5em,between origins},text height=1.5ex, text depth=.25ex]
      {P_1: & I_1 & D_1 & A_1 & I_1 \\
            & R &R{\cdotp}D_2&R{\cdotp}A_1& R \\
       P_2: & I_2 & D_2 & A_2 & I_2 \\};
      \path[morphism]
      (m-1-3) edge node[above]{$n_1$} (m-1-2)
      (m-1-3) edge node[above]{$d_1$} (m-1-4)
      (m-1-4) edge node[above]{$a_1$} (m-1-5)
      (m-3-3) edge node[below]{$n_2$} (m-3-2)
      (m-3-3) edge node[below]{$d_2$} (m-3-4)
      (m-3-4) edge node[below]{$a_2$} (m-3-5);
      \path[morphism,densely dashed]
      (m-2-2) edge node[left]{$r_1$} (m-1-2)
      (m-2-2) edge node[left]{$r_2$} (m-3-2)
      (m-2-5) edge node[right]{$r_1$} (m-1-5)
      (m-2-5) edge node[right]{$r_2$} (m-3-5)
      (m-2-3) edge node[above]{$\gamma$} (m-2-2)
      (m-2-3) edge (m-2-4)
      (m-2-4) edge (m-2-5)
      (m-2-3) edge node[right]{$\beta$} (m-1-3)
      (m-2-3) edge (m-3-3)
      (m-2-4) edge (m-1-4)
      (m-2-4) edge node[right]{$\alpha$} (m-3-4);
      \pullback[.6cm]{m-1-4}{m-2-4}{m-2-5}{draw,-};
      \pullback[.6cm]{m-3-3}{m-2-3}{m-2-4}{draw,-};
    \end{tikzpicture}
    \ .
  \]
\end{defi}
We can internalize the notion of simulation using the language
of dependent types:
\begin{prop}\label{prop:TTsim}
  A simulation from~$P_1$ to~$P_2$ is given by: (refer to
  Definition~\ref{def:simulationAsNatTransformation})
  \begin{enumerate}
    \item $i_1:I_1, i_2:I_2 \vdash R(i_1,i_2)$ for the span,

    \item $i_1, i_2, r:R({i_1,i_2}), a_1:A_1(i_1) \vdash \alpha(i_1,i_2,r,a_1)
      : A_2(i_2)$ for the morphism~$\alpha$,

    \item $i_1, i_2, r, a_1, d_2:D_2\big(i_2,\alpha(...)\big) \vdash
      \beta(i_1,i_2,r,a_1,d_2) : D_1(i_1,a_1)$ for the morphism~$\beta$,

    \item $i_1, i_2, r, a_1, d_2:D_2\big(i_2,\alpha(...)\big) \vdash
      \gamma(...,d_2) : R\big(i_2[a_1/\beta(...)],i_2[\alpha(...)/d_2]\big)$ for
      the morphism~$\gamma$,

  \end{enumerate}
  Putting all this together and rewriting it more concisely, a simulation is given by:
  \begin{itemize}
    \item $i_1:I_1, i_2:I_2 \vdash R$ for the span,
    \item $i_1, i_2 \vdash \pi : \prod_{a_1}
      \sum_{a_2} \prod_{d_2} \sum_{d_1} \
      R\big(i_1[a_1/d_1],i_2[a_2/d_2]\big)$, where for~$k=1,2$ the types are~$a_k:A_k(i_k)$
      and~$d_k:D_k(i_k,a_k)$.
  \end{itemize}
\end{prop}
\begin{proof}
  We'll only show the beginning in order to give a taste of the manipulations
  involved.
  Let's first fix some notation: the simulation is given by the diagram
  \[
    \begin{tikzpicture}[baseline=(m-3-2.base)]
      \matrix (m) [matrix of math nodes, column sep={4em,between origins},
      row sep={4em,between origins}, text height=1.5ex,text depth=.25ex]
      { I_1 & D_1 & A_1 & I_1 \\
          R & \cdotp & A_1{\cdotp}R& R \\
        I_2 & D_2 & A_2 & I_2 \\};
      \path[morphism]
      (m-1-2) edge node[above]{$n_1$} (m-1-1)
      (m-1-2) edge node[above]{$d_1$} (m-1-3)
      (m-1-3) edge node[above]{$a_1$} (m-1-4)
      (m-3-2) edge node[below]{$n_2$} (m-3-1)
      (m-3-2) edge node[below]{$d_2$} (m-3-3)
      (m-3-3) edge node[below]{$a_2$} (m-3-4);
      \path[morphism,densely dashed]
      (m-2-1) edge node[left]{$s$} (m-1-1)
      (m-2-1) edge node[left]{$t$} (m-3-1)
      (m-2-4) edge node[right]{$s$} (m-1-4)
      (m-2-4) edge node[right]{$t$} (m-3-4)
      (m-2-3) edge node[left]{$y$} (m-1-3)
      (m-2-3) edge node[below]{$x$} (m-2-4)
      (m-2-2) edge node[above]{$\gamma$} (m-2-1)
      (m-2-2) edge (m-2-3)
      (m-2-3) edge (m-2-4)
      (m-2-2) edge node[right]{$\beta$} (m-1-2)
      (m-2-2) edge (m-3-2)
      (m-2-3) edge (m-1-3)
      (m-2-3) edge node[left]{$\alpha$} (m-3-3);
      \pullback[.5cm]{m-1-3}{m-2-3}{m-2-4}{draw,-};
      \pullback[.5cm]{m-3-2}{m-2-2}{m-2-3}{draw,-};
    \end{tikzpicture}
    \ .
  \]
  The following pullbacks will be useful in the sequel:
  \[
    \begin{tikzpicture}[baseline=(m-2-2.base)]
      \matrix (m) [matrix of math nodes, column sep={5em,between origins}, row
      sep={5em,between origins},text height=1.5ex, text depth=.25ex]
      { A_1{\cdotp}R& A_1\times I_2 & A_1 \\
             R      & I_1\times I_2 & I_1 \\};
      \path[morphism]
      (m-1-1) edge node[above]{$\langle y,tx\rangle$} (m-1-2)
      (m-1-2) edge node[left]{$a_1{\times}1$} (m-2-2)
      (m-1-1) edge node[left]{$x$} (m-2-1)
      (m-2-1) edge node[below]{$\langle s,t\rangle$} (m-2-2)
      (m-1-2) edge node[above]{$\pi_1$} (m-1-3)
      (m-1-3) edge node[right]{$a_1$} (m-2-3)
      (m-2-2) edge node[below]{$\pi_1$} (m-2-3);
      \pullback[.5cm]{m-2-1}{m-1-1}{m-1-2}{draw,-}
      \pullback[.5cm]{m-2-2}{m-1-2}{m-1-3}{draw,-}
    \end{tikzpicture}
    \qquad
%
    \begin{tikzpicture}[baseline=(m-2-2.base)]
      \matrix (m) [matrix of math nodes, column sep={5em,between origins}, row
      sep={5em,between origins},text height=1.5ex, text depth=.25ex]
      { A_1{\cdotp}R{\cdotp}A_2 & A_1{\cdotp}R & A_2 \\
                   A_1{\cdotp}R & R          & I_2 \\};
      \path[morphism]
      (m-1-1) edge node[above]{$u$} (m-1-2)
      (m-1-2) edge node[above]{$v$} (m-1-3)
      (m-1-2) edge (m-2-2)
      (m-1-1) edge node[left]{$f$} (m-2-1)
      (m-2-1) edge node[below]{$x$} (m-2-2)
      (m-1-3) edge node[right]{$a_2$} (m-2-3)
      (m-2-2) edge node[below]{$t$} (m-2-3);
      \pullback[.5cm]{m-2-1}{m-1-1}{m-1-2}{draw,-}
      \pullback[.5cm]{m-2-2}{m-1-2}{m-1-3}{draw,-}
    \end{tikzpicture}
    \ .
  \]
  Interpreting the above types in~$\C$ gives:
  \begin{enumerate}
    \item $\vdash I_1$ is ``$I_1$'',
    \item $i_1:I_1\vdash I_2$ is ``$\pi_1:I_1\times I_2 \to I_1$'',
    \item $i_1,i_2:I_2 \vdash R(i_1,i_2)$ is ``$\langle s,t\rangle : R\to I_1\times I_2$'',
    \item $i_1,i_2,r:R(i_1,i_2) \vdash A_1(i_1)$ is ``$x:A_1{\cdotp}R\to R$'',
    \item $i_1,i_2,r,a_1:A_1(i_1) \vdash A_2(i_2)$ is ``$f : A_1{\cdotp}R{\cdotp}A_2 \to
      A_1{\cdotp}R$''.
  \end{enumerate}
  The term $i_1, i_2, r, a_1 \vdash \alpha(i_1,i_2,r,a_1) : A_2(i_2)$ thus
  corresponds to a section~$\varphi$ of~$f$.

  The sections of~$f$ are in 1-1 correspondence with the
  morphisms~$\alpha:A_1{\cdotp}R\to A_2$ s.t.~$tx=b\alpha$: given such
  an~$\alpha$, construct~$\gamma$ as the mediating arrow in
  \[
    \begin{tikzpicture}[baseline=(m-3-3.base)]
      \matrix (m) [matrix of math nodes, column sep={4em,between origins}, row
      sep={4em,between origins},text height=1.5ex, text depth=.25ex]
      { A_1{\cdotp}R &            &        &   \\
                     & A_1{\cdotp}R{\cdotp}A_2 & \cdotp & A_2 \\
                     & A_1{\cdotp}R          & \cdotp & I_2 \\};
      \path[morphism]
      (m-2-2) edge node[above]{$u$} (m-2-3)
      (m-2-3) edge node[above]{$v$} (m-2-4)
      (m-2-2) edge node[left]{$f$} (m-3-2)
      (m-3-2) edge node[below]{$x$} (m-3-3)
      (m-2-4) edge node[right]{$b$} (m-3-4)
      (m-3-3) edge node[below]{$t$} (m-3-4);
      \path[morphism,bend left=25]
      (m-1-1) edge node[above]{$\alpha$} (m-2-4);
      \path[morphism,bend left=-35]
      (m-1-1) edge node[left]{$1$} (m-3-2) ;
      \path[morphism,dashed]
      (m-1-1) edge node[above right]{$\gamma$} (m-2-2);
      \pullback[.5cm]{m-3-2}{m-2-2}{m-2-3}{draw,-}
    \end{tikzpicture}
    \ .
  \]
  The lower triangle shows that this~$\gamma$ is a section.

  The inverse of this construction is given by~$\gamma\mapsto vu\gamma$. Because of
  the upper triangle above, this is indeed a left inverse of the ``mediating
  arrow'' transformation.
  This is also a right inverse: let~$\gamma$ be s.t.~$f\gamma=\id$. We
  have~$bvu\gamma = txf\gamma = tx$ and so, if we use~$\alpha \eqdef vu\gamma$ in the above
  diagram, the mediating arrow will necessarily be~$\gamma$.

  The rest is similar.
\end{proof}

\medbreak
Defining the composition of simulations using type theory isn't too difficult,
but here is the diagrammatic representation of such a composition:
\[
  \begin{tikzpicture}[baseline=(m-4-8.base)]
    \matrix (m) [matrix of math nodes, column sep={2em,between origins}, row
    sep={2.5em,between origins}, text height=1.5ex,text depth=0.25ex]
    {   &   &I_1&   &   &   &D_1&   &   &   &A_1&   &   &   &I_1&   \\
        &   &   &   &   &   &   &   &   &   &   &   &   &   &   &   \\
        &   & R &   &   &   & Y &   &   &   & X &   &   &   & R &   \\
        &   &   &I_2&   &   &   &D_2&   &   &   &A_2&   &   &   &I_2\\
      T &   &   &   & V &   &   &   & U &   &   &   & T &   &   &   \\
        &R' &   &   &   &Y' &   &   &   &X' &   &   &   &R' &   &   \\
        &   &   &   &   &   &   &   &   &   &   &   &   &   &   &   \\
        &I_3&   &   &   &D_3&   &   &   &A_3&   &   &   &I_3&   &   \\};
    \foreach \i in {1,3}{
      \path[morphism]
      (m-\i-7) edge (m-\i-3)
      (m-\i-7) edge (m-\i-11)
      (m-\i-11) edge (m-\i-15);
    }
    \foreach \i in {6,8}{
      \path[morphism]
      (m-\i-6) edge (m-\i-2)
      (m-\i-6) edge (m-\i-10)
      (m-\i-10) edge (m-\i-14);
    }
    \path[morphism]
    (m-3-3) edge (m-1-3)
    (m-3-7) edge (m-1-7)
    (m-3-11) edge (m-1-11)
    (m-3-15) edge (m-1-15);
    \path[morphism]
    (m-6-2) edge (m-8-2)
    (m-6-6) edge (m-8-6)
    (m-6-10) edge (m-8-10)
    (m-6-14) edge (m-8-14);
    \path[morphism]
    (m-3-3) edge (m-4-4)
    (m-3-7) edge (m-4-8)
    (m-3-11) edge (m-4-12)
    (m-3-15) edge (m-4-16);
    \path[morphism,densely dashed]
    (m-5-1) edge (m-3-3)
    (m-5-5) edge (m-3-7)
    (m-5-9) edge (m-3-11)
    (m-5-13) edge (m-3-15);
    \path[morphism,densely dashed]
    (m-5-1) edge (m-6-2)
    (m-5-5) edge (m-6-6)
    (m-5-9) edge (m-6-10)
    (m-5-13) edge (m-6-14);
    \path[morphism,densely dashed]
    (m-5-5) edge (m-5-1)
    (m-5-5) edge (m-5-9)
    (m-5-9) edge (m-5-13);
    \path[morphism]
    (m-4-8) edge[crossover] (m-4-4)
    (m-4-8) edge[crossover] (m-4-12)
    (m-4-12) edge[crossover] (m-4-16);
    \path[morphism]
    (m-6-2) edge[crossover] (m-4-4)
    (m-6-6) edge[crossover] (m-4-8)
    (m-6-10) edge[crossover] (m-4-12)
    (m-6-14) edge (m-4-16);

    \pullback{m-1-11}{m-3-11}{m-3-15}{draw,-};
    \pullback{m-3-11}{m-3-7}{m-4-8}{draw,-};
    \pullback{m-6-10}{m-6-6}{m-8-6}{draw,-};
    \pullback{m-6-14}{m-6-10}{m-4-12}{draw,-};

    \pullback{m-3-3}{m-5-1}{m-6-2}{crossover,draw,-}
    \pullback{m-3-7}{m-5-5}{m-6-6}{crossover,draw,-}
    \pullback{m-3-11}{m-5-9}{m-6-10}{crossover,draw,-}
    \pullback{m-3-15}{m-5-13}{m-6-14}{draw,-}
  \end{tikzpicture}
  \ .
\]
The plain arrows show the polynomials~$P_1$,~$P_2$ and~$P_3$ and the initial
two simulation diagrams.  The dashed diagonal arrows are computed from the two
simulations by pullbacks and the dashed horizontal arrows are mediating
morphisms. To show that the ``background layer'' forms a simulation from~$P_1$
to~$P_3$, we need to show that the squares~$(U,A_1,I_1,T)$ and~$(V,U,A_3,D_3)$
are pullbacks.  For~$(U,A_1,I_1,T)$, we know by the pullback lemma
that~$(U,X,I_2,R')$ is a pullback by pasting~$(U,X,A_2,X')$
and~$(X',A_2,I_2,R')$. A second application of the pullback lemma
on~$(U,X,R,T)$ and~$(T,R,I_2,R,)$ shows that~$(U,X,R,T)$ is also a pullback.
Finally, a third application shows that~$(U,A_1,I_1,T)$ is a pullback, as
expected.  The same reasoning shows that the square~$(V,U,A_3,D_3)$ is also a
pullback.

\medbreak
As in the case of spans, composition of simulations is only associative up to
isomorphism. We will thus need to consider equivalence classes of such
simulations. Two simulations
\[
  \begin{tikzpicture}[baseline=(m-2-2)]
    \matrix (m) [matrix of math nodes, column sep={3em,between origins}, row
    sep={3em,between origins},text height=1.5ex, text depth=.25ex]
    {I_1 & D_1 & A_1 & I_1 \\
       R &\cdotp&\cdotp& R \\
     I_2 & D_2 & A_2 & I_2\\};
    \path[morphism]
    (m-1-2) edge (m-1-1)
    (m-1-2) edge (m-1-3)
    (m-1-3) edge (m-1-4)
    (m-3-2) edge (m-3-1)
    (m-3-2) edge (m-3-3)
    (m-3-3) edge (m-3-4);
    \path[morphism]
    (m-2-1) edge node[left]{$r_1$} (m-1-1)
    (m-2-1) edge node[left]{$r_2$} (m-3-1)
    (m-2-4) edge node[right]{$r_1$} (m-1-4)
    (m-2-4) edge node[right]{$r_2$} (m-3-4)
    (m-2-2) edge node[below]{$\gamma$} (m-2-1)
    (m-2-2) edge (m-2-3)
    (m-2-3) edge (m-2-4)
    (m-2-2) edge node[left]{$\beta$}  (m-1-2)
    (m-2-2) edge (m-3-2)
    (m-2-3) edge (m-1-3)
    (m-2-3) edge node[right]{$\alpha$} (m-3-3);
    \pullback[.4cm]{m-1-3}{m-2-3}{m-2-4}{draw,-};
    \pullback[.4cm]{m-3-2}{m-2-2}{m-2-3}{draw,-};
  \end{tikzpicture}
  \quad\hbox{and}\quad
  \begin{tikzpicture}[baseline=(m-2-2)]
    \matrix (m) [matrix of math nodes, column sep={3em,between origins}, row
    sep={3em,between origins},text height=1.5ex, text depth=.25ex]
    {I_1 & D_1 & A_1 & I_1 \\
      R' &\cdotp&\cdotp& R' \\
     I_2 & D_2 & A_2 & I_2\\};
    \path[morphism]
    (m-1-2) edge (m-1-1)
    (m-1-2) edge (m-1-3)
    (m-1-3) edge (m-1-4)
    (m-3-2) edge (m-3-1)
    (m-3-2) edge (m-3-3)
    (m-3-3) edge (m-3-4);
    \path[morphism]
    (m-2-1) edge node[left]{$r'_1$} (m-1-1)
    (m-2-1) edge node[left]{$r'_2$} (m-3-1)
    (m-2-4) edge node[right]{$r'_1$} (m-1-4)
    (m-2-4) edge node[right]{$r'_2$} (m-3-4)
    (m-2-2) edge node[below]{$\gamma'$} (m-2-1)
    (m-2-2) edge (m-2-3)
    (m-2-3) edge (m-2-4)
    (m-2-2) edge node[left]{$\beta'$}  (m-1-2)
    (m-2-2) edge (m-3-2)
    (m-2-3) edge (m-1-3)
    (m-2-3) edge node[right]{$\alpha'$} (m-3-3);
    \pullback[.4cm]{m-1-3}{m-2-3}{m-2-4}{draw,-};
    \pullback[.4cm]{m-3-2}{m-2-2}{m-2-3}{draw,-};
  \end{tikzpicture}
\]
are equivalent if there are isomorphisms making the following diagram
commute:
\[
  \begin{tikzpicture}[baseline=(m-2-4.base)]
    \matrix (m) [matrix of math nodes, column sep={2em,between origins}, row
    sep={4em,between origins}, text height=1.5ex,text depth=0.25ex]
    {   &   &I_1&   &   &   &D_1&   &   &   &A_1&   &   &   &I_1&   \\
        &   &   & R'&   &   &  &\cdot&  &   &  &\cdot&  &   &   & R'\\
     R  &   &   &  &\cdot&  &   &  &\cdot&  &   &   &R  &   &   &   \\
        &I_1&   &   &   &D_2&   &   &   &A_2&   &   &   &I_2&   &   \\};

    \pullback{m-1-11}{m-2-12}{m-2-16}{draw,-}
    \pullback{m-2-12}{m-2-8}{m-4-6}{draw,-}
    \pullback{m-1-11}{m-3-9}{m-3-13}{draw,-}
    \pullback{m-3-9}{m-3-5}{m-4-6}{draw,-}
    \draw[draw=white,line width=6pt]
      (m-2-12) -- (m-3-9);
    \path[morphism,densely dashed]
      (m-3-1) edge node[over,sloped]{$\ \sim\ $} node[above,sloped]{$\sigma_R$} (m-2-4)
      (m-3-5) edge node[over,sloped]{$\ \sim\ $} node[above,sloped]{$\sigma_D$} (m-2-8)
      (m-3-9) edge node[over,sloped]{$\ \sim\ $} node[above,sloped]{$\sigma_A$} (m-2-12)
      (m-3-13) edge node[over,sloped]{$\ \sim\ $} node[above,sloped]{$\sigma_R$} (m-2-16);

    \path[morphism]
      (m-1-7) edge (m-1-3)
      (m-1-7) edge (m-1-11)
      (m-1-11) edge (m-1-15);
    \path[morphism]
      (m-4-6) edge (m-4-2)
      (m-4-6) edge (m-4-10)
      (m-4-10) edge (m-4-14);
    \path[morphism]
      (m-2-4) edge node[over]{$r'_1$} (m-1-3)
      (m-2-8) edge node[over]{$\beta$} (m-1-7)
      (m-2-12) edge (m-1-11)
      (m-2-16) edge node[over]{$r'_1$} (m-1-15);
    \path[morphism,densely dotted]
      (m-3-1) edge node[over]{$r_1$} (m-1-3)
      (m-3-5) edge node[near end,over]{$\beta$} (m-1-7)
      (m-3-9) edge (m-1-11)
      (m-3-13) edge node[near start,over]{$r_1$} (m-1-15);
    \path[morphism,densely dotted]
      (m-3-1) edge node[over]{$r_2$} (m-4-2)
      (m-3-5) edge (m-4-6)
      (m-3-9) edge node[over]{$\alpha$} (m-4-10)
      (m-3-13) edge node[over]{$r_2$} (m-4-14);
    \path[morphism,densely dotted]
      (m-3-5) edge node[near start,over]{$\gamma$} (m-3-1)
      (m-3-5) edge (m-3-9)
      (m-3-9) edge (m-3-13);
    \path[morphism]
      (m-2-8) edge[crossover] node[near end,over]{$\gamma'$} (m-2-4)
      (m-2-8) edge[crossover] (m-2-12)
      (m-2-12) edge[crossover] (m-2-16);
    \path[morphism]
      (m-2-4) edge[crossover] node[near start,over]{$r'_2$} (m-4-2)
      (m-2-8) edge[crossover] (m-4-6)
      (m-2-12) edge[crossover] node[near start,over]{$\alpha'$} (m-4-10)
      (m-2-16) edge node[over]{$r'_2$} (m-4-14);

  \end{tikzpicture}
  \ .
\]
Such isomorphisms are in fact induced by a single span isomorphism
\[
  \begin{tikzpicture}[baseline=(m-2-1)]
    \matrix (m) [matrix of math nodes, column sep={4.5em,between origins}, row
    sep={3em,between origins},text height=1.5ex, text depth=.25ex]
    {     & R  &     \\
      I_1 &    & I_2 \\
          & R' &     \\};
    \path[morphism]
    (m-1-2) edge node[above left]{$r_1$} (m-2-1)
    (m-1-2) edge node[above right]{$r_2$} (m-2-3)
    (m-3-2) edge node[below left]{$r'_1$} (m-2-1)
    (m-3-2) edge node[below right]{$r'_2$} (m-2-3);
    \path[morphism,densely dashed]
    (m-1-2) edge node[over,sloped]{$\ \sim\ $} node[above,sloped]{$\sigma_R$} (m-3-2);
  \end{tikzpicture}
  \ .
\]
We can now define:
\begin{defi}
  $\PE_\C$ is the category of polynomial endofunctors diagrams (simply called
  polynomials) and equivalence classes of simulation diagrams as in
  Definition~\ref{def:simulationAsNatTransformation}.
\end{defi}


\subsection{Tensor Product and SMCC Structure} 

The tensor is just a ``pointwise cartesian product'':
\begin{defi}
  The polynomial~$P_1\Tensor P_2$ is defined as
  \[
    P_1\Tensor P_2 \quad\eqdef\quad
    \begin{tikzpicture}[baseline=(m-1-1.base)]
      \matrix (m) [matrix of math nodes, column sep=3.5em]
      { I_1\times I_2 & D_1\times D_2 & A_1\times A_2 & I_1\times I_2 \\ };
      \path[morphism]
      (m-1-2) edge node[above]{$n_1{\times}n_2$} (m-1-1)
      (m-1-2) edge node[above]{$d_1{\times}d_2$} (m-1-3)
      (m-1-3) edge node[above]{$a_1{\times}a_2$} (m-1-4);
    \end{tikzpicture}
    \ .
  \]
\end{defi}
In terms of games, Alfred and Dominic play synchronously in the two
games~$P_1$ and~$P_2$ at the same time. Composition and simulation diagrams
lift pointwise, making it straightforward to check that
\begin{lem}\label{lem:tensor}
  $\BLANK\Tensor\BLANK$ is a bifunctor in the category~$\PE_\C$. It
  has a neutral element given by~$\One\leftarrow\One\to\One\to\One$
  where~$\One$ is the terminal element of~$\C$.
\end{lem}

We will now prove that:
\begin{prop}\label{prop:linear}
  The category~$\PE_\C$ with~$\Tensor$ is symmetric monoidal closed, i.e.,
  there is a functor~$\BLANK\Linear\BLANK$
  from~$\PE_\C^{\mathrm{op}}\times\PE_\C$ to~$\PE_\C$ with an adjunction
  \[
    \PE_\C[P_1\Tensor P_2\ ,\ P_3]
    \quad\iso\quad
    \PE_\C[P_1\ ,\ P_2\Linear P_3]
    \ ,
  \]
  natural in~$P_1$ and~$P_3$.
\end{prop}
Let's start by giving a definition of~$P_2\Linear P_3$ using the internal
language of LCCC. A purely diagrammatic definition of~$P_2\Linear P_3$ will
follow.
\begin{defi}\label{defn:linear}
  The polynomial~$P_2\Linear P_3$ is defined as:
  \begin{enumerate}
    \item $\vdash I_2\times I_3$,
    \item $i_2, i_3 \vdash \sum_{f:A_2(i_2) \to A_3(i_3)} \prod_{a_2:A_2(i_2)}
      D_3\big(i_3,f(a_2)\big) \to D_2(i_2,a_2)$,
    \item $i_2, i_3\ ,\ f, \varphi \vdash \sum_{a_2:A_2(i_2)}
      D_3\big(i_3,{f(a_2)}\big)$,
    \item $i_2, i_3\ ,\ f, \varphi\ ,\ a_2, d_3 \vdash
      \big(i_2[a_2/\varphi(a_2)(d_3)],i_3[f(a_2)/d_3]\big):I_2\times I_3$,
  \end{enumerate}
  where the types of variables are as follows: $f$ is of type~$A_2(i_2)\to
  A_3(i_3)$, $\varphi$ is of type~$\prod_{a_2:A_2(i_2)}
  D_3\big(i_3,f(a_2)\big) {\to} D_2(i_2,a_2)$, $a_2$ is of type~$A_2(i_2)$ and
  $d_3$ is of type~$D_3\big(i_3,f(a_2)\big)$.
\end{defi}

\begin{proof}[Proof of proposition~\ref{prop:linear}]
  There is a canonical natural isomorphism
  \[
    \Span_\C[I_1,I_2\times I_3]
    \quad \iso \quad
    \Span_\C[I_1\times I_2, I_3]
    \quad \iso \quad
    \Sl\C{I_1{\times}I_2{\times}I_3}
  \]
  and we use it implicitly. In order to show the adjunction, we need to find a
  natural isomorphism between
  \[
    \prod_{a_1,a_2}\quad
    \sum_{a_3}\quad
    \prod_{d_3}\quad
    \sum_{d_1,d_2}\quad
    R\big(i_1[a_1/d_1]\ ,\ i_2[a_2/d_2]\ ,\ i_3[a_3/d_3]\big)
  \]
  meaning that~$R$ is a simulation from~$P_1\Tensor P_2$ to~$P_3$ and
  \[
    \prod_{a_1}\quad
    \sum_{f,\varphi}\quad
    \prod_{a_2,d_3}\quad
    \sum_{d_1}\quad
    R\big(i_1[a_1/d_1]\ ,\ i_2[a_2/\varphi(a_2)(d_3)]\ ,\ i_3[f(a_2)/d_3]\big)
    \ ,
  \]
  meaning that~$R$ is a simulation from~$P_1$ to~$P_2\Linear P_3$. The types
  are as follows:
  \begin{itemize}
    \item $a_k:A_k(i_k)$ for $k=1,2,3$,
    \item $d_k:D_k(a_k)$ for $k=1,2,3$,
    \item $f:{A_2}({i_2}) \to {A_3}({i_3})$,
    \item $\varphi : \prod_{a_2} {D_3}\big({f(a_2)}\big) \to {D_2}({a_2})$.
  \end{itemize}
  This is just a sequence of ``distributivity'' (type theoretic axiom of
  choice, page~\pageref{diag:AC}) and obvious isomorphisms changing the order
  of independent variables:
  \begin{enumerate}
    \item from $\prod_{a_2} \sum_{a_3}$ to $\sum_f\prod_{a_2}$, to get
      $\prod_{a_1}\sum_f\prod_{a_2}\prod_{d_3}\sum_{d_1,d_2}\ \dots$

    \item from $\prod_{d_3}\sum_{d_2}$ to $\sum_{g}\prod_{d_3}$, to get
      $\prod_{a_1}\sum_f\prod_{a_2}\sum_{g}\prod_{d_3}\sum_{d_1}\ \dots$

    \item from $\prod_{a_2}\sum_{g}$ to $\sum_{\varphi}\prod_{a_2}$, to get
      $\prod_{a_1}\sum_f\sum_{\varphi}\prod_{a_2}\prod_{d_3}\sum_{d_1}\ \dots$

  \end{enumerate}
  The ``\dots'' use exactly the appropriate substitutions to make the last
  line into what was needed: $a_3 \eqdef f(a_2)$, $d_2 \eqdef g(d_3)$
  and~$g \eqdef \varphi(a_2)$.
\end{proof}

For completeness, here is the diagrammatic definition of~$P_2\Linear P_3$:
\begin{lem}\label{lem:linearDiag}
  The polynomial~$P_2\Linear P_3$ is given by
  \[
    \begin{tikzpicture}[baseline=(m-5-9)]
      \matrix (m) [matrix of math nodes, column sep={3em,between origins}, row sep={3em,between origins},text height=1.5ex, text depth=.25ex]
      {
        &   &   &   & \bullet &   & \cdotp &   & \bullet \\
        &   &   & \cdotp &   & \cdotp &   & {\scriptstyle(iii)}   &   \\
        &   & \cdotp &   & \cdotp &   & \cdotp &   & \cdotp \\
        &D_2{\times}D_3&   &A_2{\times}D_3&   &A_2{\times}A_3&   &
            {\scriptstyle(i)}  &   \\
                 I_2{\times}I_3&&  &   &   &   &A_2{\times}I_3&   & I_2{\times}I_3 \\
      };
      \draw ($(m-4-4)!.5!(m-3-5)$) node[auto]{$\scriptstyle(ii)$};
      \path[morphism]
      (m-1-5) edge node[above]{$d'''_2$} (m-1-7)
      (m-1-7) edge node[above]{$a''_2$} (m-1-9)
      (m-2-4) edge node[above]{$d''_2$} (m-2-6)
      (m-3-3) edge node[above]{$d'_2$} (m-3-5)
      (m-3-7) edge node[above]{$a'_2$} (m-3-9)
      (m-4-2) edge node[below]{$d_2{\times}1$} (m-4-4)
      (m-4-4) edge node[below]{$1{\times}d_3$} (m-4-6)
      (m-5-7) edge node[below]{$a_2{\times}1$} (m-5-9)
      (m-1-7) edge (m-3-7)
      (m-3-7) edge (m-5-7)
      (m-1-9) edge node[right]{$\varphi$} (m-3-9)
      (m-3-9) edge node[right]{$f$} (m-5-9)
      (m-1-5) edge node[above left]{$e$} (m-2-4)
      (m-1-7) edge node[over]{$\epsilon$} (m-2-6)
      (m-2-4) edge node[above left]{$e''$} (m-3-3)
      (m-2-6) edge node[above left]{$e'$} (m-3-5)
      (m-3-3) edge node[over]{$\epsilon$} (m-4-2)
      (m-3-3) edge (m-4-4)
      (m-3-5) edge node[above right]{$g$} (m-4-6)
      (m-2-6) edge node[above right]{$g'$} (m-3-7)
      (m-3-7) edge node[over]{$\epsilon$} (m-4-6)
      (m-4-6) edge node[below left]{$1{\times}a_3$} (m-5-7)
      (m-4-2) edge node[above left]{$n_2{\times}n_3$} (m-5-1);
      \pullback[.6cm]{m-2-4}{m-1-5}{m-1-7}{draw,-};
      \pullback[.6cm]{m-3-3}{m-2-4}{m-2-6}{draw,-};
      \pullback[.6cm]{m-4-4}{m-3-3}{m-3-5}{draw,-};
      \pullback[.6cm]{m-3-5}{m-2-6}{m-3-7}{draw,-};
      \pullback[.6cm]{m-3-7}{m-1-7}{m-1-9}{draw,-};
      \pullback[.6cm]{m-5-7}{m-3-7}{m-3-9}{draw,-};
    \end{tikzpicture}
    \ .
  \]
  where the~$1$s are appropriate identities and squares~$(i)$, $(ii)$
  and~$(iii)$ are distributivity squares as in diagram~$(\ref{diag:distr})$,
  i.e., the~$\epsilon$s are appropriate counits of the
  adjunction~$\Delta_\BLANK \dashv \Pi_\BLANK$.
\end{lem}
It is not too difficult to show that this corresponds to the type theoretic
definition. However, an elegant and direct proof that this is indeed the
right adjoint for~$\BLANK\Tensor P_2$ is still to be found. One easy thing
is the following (compare it with the second part of
Proposition~\ref{prop:TTsim})
\begin{lem}\label{lem:ExtensionLinear}
  The extension of~$P_2\Linear P_3$ is
  \begin{eqnarray*}
    \Sem{P_2\Linear P_3}
    &\quad=\quad&
    \Pi_{a_2{\times}1} \, \Sigma_{1{\times}a_3} \, \Pi_{1{\times}d_3} \, \Sigma_{d_2{\times}1} \, \Delta_{n_2{\times}n_3}\\
    &\quad=\quad&
    \Pi_{a_2{\times}1} \, \Sem{\One \Tensor P_3}\, \Sigma_{d_2{\times}1} \, \Delta_{n_2{\times}1}
    \ .
  \end{eqnarray*}
\end{lem}
\begin{proof}
  This is just a rewriting of the definition using Beck-Chevalley and
  distributivity as appropriate. Using the notation from
  Lemma~\ref{lem:linearDiag}:
  {\allowdisplaybreaks
    \begin{eqnarray*}
      \Sem{P_2\Linear P_3}
      &\quad=\Big.\quad&
      \Sigma_f \, \Sigma_\varphi \, \Pi_{a''_2} \, \underline{\Pi_{d'''_2} \, \Delta_e} \, \Delta_{e''} \, \Delta_\epsilon \, \Delta_{n_2{\times}n_3} \, \\
      &\quad=\Big.\quad&
      \Sigma_f \, \underline{\Sigma_\varphi \, \Pi_{a''_2} \, \Delta_\epsilon} \, \Pi_{d''_2} \, \Delta_{e''} \, \Delta_\epsilon \, \Delta_{n_2{\times}n_3}\\
      &\quad=\Big.\quad&
      \Sigma_f \, \Pi_{a'_2} \, \Sigma_{g'} \, \underline{\Pi_{d''_2} \, \Delta_{e''}} \, \Delta_\epsilon \, \Delta_{n_2{\times}n_3} \, \\
      &\quad=\Big.\quad&
      \Sigma_f \, \Pi_{a'_2} \, \underline{\Sigma_{g'} \, \Delta_{e'}} \, \Pi_{d'_2} \, \Delta_\epsilon \, \Delta_{n_2{\times}n_3}\\
      &\quad=\Big.\quad&
      \underline{\Sigma_f \, \Pi_{a'_2} \, \Delta_\epsilon} \, \underline{\Sigma_g \, \Pi_{d'_2} \, \Delta_\epsilon} \, \Delta_{n_2{\times}n_3}\\
      &\quad=\Big.\quad&
      \Pi_{a_2{\times}1} \, \Sigma_{1{\times}a_3} \, \Pi_{1{\times}d_3} \, \Sigma_{d_2{\times}1} \, \Delta_{n_2{\times}n_3}
      \ .
    \end{eqnarray*}
    The second equality follows from a Beck-Chevalley
    isomorphism:
    \[
      \Sigma_{d_2{\times}1} \Delta_{1{\times}n_3} \quad = \quad
      \Delta_{1{\times}n_3} \Sigma_{d_2{\times}1}
      \ .
    \]
  }
\end{proof}

\subsection{Linear Negation} 
\label{sec:dual}

Of particular interest is the dual of~$P$:~$P^\bot \eqdef P\Linear\One$ where~$\One$
is the neutral element for~$\Tensor$.
\[
  \begin{tikzpicture}[baseline=(m-1-4.base)]
    \matrix (m) [matrix of math nodes, column sep=2em, row sep=2em]
    { \One & \One & \One & \One \\};
    \path[morphism]
    (m-1-2) edge node[above]{$1$} (m-1-1)
    (m-1-2) edge node[above]{$1$} (m-1-3)
    (m-1-3) edge node[above]{$1$} (m-1-4);
  \end{tikzpicture}
  \ .
\]
In the denotational model for intuitionistic linear logic, this polynomial is
the most natural choice for the~``$\bot$'' object. The polynomial~$P^\bot$ is
thus the ``linear negation'' of~$P$.
By simplifying definition~\ref{defn:linear} in this case, we obtain:
\begin{defi}
  Given a polynomial~$P = (I\leftarrow D\to A\to I)$, the polynomial~$P^\bot$ is
  defined by:
  \begin{enumerate}
    \item $\vdash I$,
    \item $i:I \vdash \prod_{a:A(i)} D(i,a)$,
    \item $i:I, f: \prod_{a:A(i)} D(i,a) \vdash A(i)$,
    \item $i:I, f: \prod_{a:A(i)} D(i,a), a:A(i)\vdash n\big(i,a,f(a)\big):I$.
  \end{enumerate}
\end{defi}
Note that this negation isn't involutive.

The dual~$G^\bot$ of a game~$G$ is rather different from the usual operation
consisting of interchanging the players as is done in games semantics. The
main property is that a strategy (for Alfred) in~$G^\bot$ is exactly a
strategy for Dominic in~$G$.
%

%
Simplifying the diagrammatic definition of~$\Linear$
(Lemma~\ref{lem:linearDiag}), we find that~$P^\bot$ is
\[
  \begin{tikzpicture}
    \matrix (m) [matrix of math nodes, column sep={3em,between origins}, row sep={3em,between origins},text height=1.5ex, text depth=.25ex]
    {   &   &   & \bullet &   & \bullet &   \\
      I &   & D &         & A &         & I \\};
    \draw ($(m-1-6)!.5!(m-2-5)$) node[auto]{$\scriptstyle(i)$};
    \path[morphism]
    (m-2-3) edge node[below]{$n$} (m-2-1)
    (m-1-4) edge node[over]{$\epsilon$} (m-2-3)
    (m-1-4) edge (m-2-5)
    (m-2-3) edge node[below]{$d$} (m-2-5)
    (m-1-4) edge (m-1-6)
    (m-2-5) edge node[below]{$a$} (m-2-7)
    (m-1-6) edge (m-2-7);
    \pullback{m-2-5}{m-1-4}{m-1-6}{draw,-};
  \end{tikzpicture}
\]
where square~$(i)$ is a distributivity square. The middle arrow is thus of the
form~$\Delta_{\BLANK}a$.
It is worth noting that for any polynomial~$P$, the polynomial~$P^\bot$
has the form of ``simultaneous games'' described on
page~\pageref{rk:HylandGames}. 


\section{Additive and Exponential Structure}  
\label{section:ISinSet}

To go further, we will need more structure from~$\C$. To avoid spelling out the exact
requirements (extensivity, existence of certain colimits etc.), we now work in
the category of sets and functions~$\C = \Set$, where we certainly have all we
need.

\subsection{Enriched Structure} 
\label{sub:enrichedStructure}

Because~$\emptyset$ is initial in~$\Set$, there is an ``empty'' simulation
between any two polynomial functors. It is given by the diagram
\[
  \begin{tikzpicture}[baseline=(m-3-2)]
    \matrix (m) [matrix of math nodes, column sep={3.5em,between origins}, row
    sep={3.5em,between origins},text height=1.5ex, text depth=.25ex]
    {P_1: & I_1 & D_1 & A_1 & I_1 \\
          & \emptyset & \emptyset & \emptyset & \emptyset \\
     P_2: & I_2 & D_2 & A_2 & I_2\\};
    \path[morphism]
    (m-1-3) edge node[above]{$n_1$} (m-1-2)
    (m-1-3) edge node[above]{$d_1$} (m-1-4)
    (m-1-4) edge node[above]{$a_1$} (m-1-5)
    (m-3-3) edge node[below]{$n_2$} (m-3-2)
    (m-3-3) edge node[below]{$d_2$} (m-3-4)
    (m-3-4) edge node[below]{$a_2$} (m-3-5);
    \path[morphism,densely dashed]
    (m-2-2) edge node[left]{$\exclamation$} (m-1-2)
    (m-2-2) edge node[left]{$\exclamation$} (m-3-2)
    (m-2-5) edge node[right]{$\exclamation$} (m-1-5)
    (m-2-5) edge node[right]{$\exclamation$} (m-3-5)
    (m-2-3) edge node[above]{$\exclamation$} (m-2-2)
    (m-2-3) edge (m-2-4)
    (m-2-4) edge (m-2-5)
    (m-2-3) edge node[right]{$\exclamation$} (m-1-3)
    (m-2-3) edge (m-3-3)
    (m-2-4) edge (m-1-4)
    (m-2-4) edge node[right]{$\exclamation$} (m-3-4);
    \pullback[.6cm]{m-1-4}{m-2-4}{m-2-5}{draw,-};
    \pullback[.6cm]{m-3-3}{m-2-3}{m-2-4}{draw,-};
  \end{tikzpicture}
  \ .
\]
This is true for any~$\C$ having an initial object~$\Zero$ because in an
LCCC, we necessarily have~$\Zero \times_{\scriptscriptstyle\!I}X \iso
\Zero$.

Moreover, each~$\Sl\Set X$ is a cocomplete category and span composition is
continuous on both sides, making the~$\Span_\Set$ enriched over (large)
sup-monoids. Because a colimit of simulations is easily made
into a simulation, this implies that~$\PE_{\Set}$ is also enriched over
sup-monoids.
\begin{prop}
  $\PE_\Set$ is enriched over large sup-monoids.
\end{prop}

\subsection{Additive Structure} 

The category~$\Set$ is extensive. This means in particular that a
slice~$f\in\Sl\Set{B+C}$ can be uniquely (up-to isomorphism) written
as~$f_B+f_C$ for some~$f_B\in\Sl\Set B$ and~$f_C\in\Sl\Set C$. In other
words,~$+$ is an equivalence of categories~$\Sl\Set{B+C}\iso\Sl\Set
B\times\Sl\Set C$.
Extensivity implies that in~$\Span_\Set$:
\begin{itemize}
  \item $\emptyset$ is a zero object,
  \item the coproduct of~$X$ and~$Y$ is given by~$X+Y$ (disjoint union),
  \item $X+Y$ is also the product of~$X$ and~$Y$,
\end{itemize}
Moreover, the infinite coproduct from~$\Set$ $\sum_{k\in K} X_k$ lifts to the
infinite coproduct and product in~$\Span_\Set$.
We have:
\begin{lem}
  The forgetful functor~$U : \PE_\Set \to \Span_\Set$ sending a polynomial
  to its domain and a simulation to its underlying span has a left and a
  right adjoint.
\end{lem}
\begin{proof}
  The object part of the adjoints~$L\dashv U\dashv R$ for set~$I$ are given by
  \[
    L(I) \ \eqdef
    \begin{tikzpicture}[baseline=(m-1-1.base)]
      \matrix (m) [matrix of math nodes, column sep=1.5em, row sep=1.5em, text
      height=1.5ex, text depth=.25ex]
      { I & \emptyset & \emptyset & I  \\};
      \path[morphism]
      (m-1-2) edge node[above]{$\exclamation$} (m-1-1)
      (m-1-2) edge node[above]{$\exclamation$} (m-1-3)
      (m-1-3) edge node[above]{$\exclamation$} (m-1-4);
    \end{tikzpicture}
    \quad\hbox{and}\quad
    R(I) \ \eqdef
    \begin{tikzpicture}[baseline=(m-1-1.base)]
      \matrix (m) [matrix of math nodes, column sep=1.5em, row sep=1.5em, text
      height=1.5ex, text depth=.25ex]
      { I & \emptyset & I & I  \\};
      \path[morphism]
      (m-1-2) edge node[above]{$\exclamation$} (m-1-1)
      (m-1-2) edge node[above]{$\exclamation$} (m-1-3)
      (m-1-3) edge node[above]{$1$} (m-1-4);
    \end{tikzpicture}
    \ .
  \]
  The rest is simple verification.
\end{proof}

Because the forgetful functor~$U$ has a left adjoint, it must preserve the
product (and so, the coproduct as well). The product of~$P_1$ and~$P_2$ is
thus a polynomial over~$I_1+I_2$.
\begin{defi}
  If~$P_1$ and~$P_2$ are polynomials over~$I_1$ and $I_2$, we write~$P_1\Plus P_2$
  for
  \[
    P_1 \Plus P_2 \quad \eqdef \quad
    \begin{tikzpicture}[baseline=(m-1-1.base)]
      \matrix (m) [matrix of math nodes, column sep=3.5em]
      { I_1+ I_2 & D_1+ D_2 & A_1+ A_2 & I_1+ I_2 \\ };
      \path[morphism]
      (m-1-2) edge node[above]{$n_1{+}n_2$} (m-1-1)
      (m-1-2) edge node[above]{$d_1{+}d_2$} (m-1-3)
      (m-1-3) edge node[above]{$a_1{+}a_2$} (m-1-4);
    \end{tikzpicture}
    \ .
  \]
  The polynomial~$\Zero$ is the unique polynomial with domain and
  codomain~$\emptyset$.
\end{defi}
We have
\begin{lem}\label{lem:PEProdCoprod}
  The bifunctor $\Plus$ is both a product and a coproduct in~$\PE_\Set$. The
  polynomial~$\Zero$ is a zero object.
\end{lem}

\begin{proof}
  The ``injections'' are given by
  \[
    \begin{tikzpicture}[baseline=(m-2-2)]
      \matrix (m) [matrix of math nodes, column sep={5em,between origins}, row
      sep={4em,between origins},text height=1.5ex, text depth=.25ex]
      { I & D & A & I \\
        I & D & A& I\\
    I{+}J & D{+}E & A{+}B & I{+}J\ .\\};
      \path[morphism]
      (m-1-2) edge node[above]{$n$} (m-1-1)
      (m-1-2) edge node[above]{$d$} (m-1-3)
      (m-1-3) edge node[above]{$a$} (m-1-4)
      (m-3-2) edge node[below]{$n{+}m$} (m-3-1)
      (m-3-2) edge node[below]{$d{+}e$} (m-3-3)
      (m-3-3) edge node[below]{$a{+}b$} (m-3-4);
      \path[morphism,densely dashed]
      (m-2-1) edge node[left]{$1$} (m-1-1)
      (m-2-1) edge node[left]{$\inl$} (m-3-1)
      (m-2-4) edge node[right]{$1$} (m-1-4)
      (m-2-4) edge node[right]{$\inl$} (m-3-4)
      (m-2-2) edge node[above]{$n$} (m-2-1)
      (m-2-2) edge (m-2-3)
      (m-2-3) edge (m-2-4)
      (m-2-2) edge node[right]{$1$} (m-1-2)
      (m-2-2) edge (m-3-2)
      (m-2-3) edge (m-1-3)
      (m-2-3) edge node[right]{$\inl$} (m-3-3);
      \pullback[.6cm]{m-1-3}{m-2-3}{m-2-4}{draw,-};
      \pullback[.6cm]{m-3-2}{m-2-2}{m-2-3}{draw,-};
    \end{tikzpicture}
  \]
  and similarly for the ``right'' injection. Because all the six squares are
  in fact pullbacks, this also defines the projections by mirroring everything horizontally.

  Now, because~$\Set$ is extensive, any simulation from~$P_1\Plus P_2$
  to~$P_3$ is of the form:
  \[
    \begin{tikzpicture}[baseline=(m-2-2)]
      \matrix (m) [matrix of math nodes, column sep={7em,between origins}, row
      sep={4em,between origins},text height=1.5ex, text depth=.25ex]
      { I_1{+}I_2 & D_1{+}D_2 & A_1{+}A_2 & I_1{+}I_2 \\
        R_1{+}R_2 &U{+}V&X{+}Y& R_1{+}R_2 \\
              I_3 & D_3 & A_3 & I_3\\};
      \draw ($(m-3-2)!.5!(m-2-3)$) node[auto]{$\scriptstyle(i)$};
      \path[morphism]
      (m-1-2) edge node[above]{$n_1{+}n_2$} (m-1-1)
      (m-1-2) edge node[above]{$d_1{+}d_2$} (m-1-3)
      (m-1-3) edge node[above]{$a_1{+}a_2$} (m-1-4)
      (m-3-2) edge node[below]{$n_3$} (m-3-1)
      (m-3-2) edge node[below]{$d_3$} (m-3-3)
      (m-3-3) edge node[below]{$a_3$} (m-3-4);
      \path[morphism,densely dashed]
      (m-2-1) edge node[left]{$r_1{+}r_2$} (m-1-1)
      (m-2-1) edge node[left]{$[r_{1,3},r_{2,3}]$} (m-3-1)
      (m-2-4) edge node[right]{$r_1{+}r_2$} (m-1-4)
      (m-2-4) edge node[right]{$[r_{1,3},r_{2,3}]$} (m-3-4)
      (m-2-2) edge node[above]{$\gamma_1{+}\gamma_2$} (m-2-1)
      (m-2-2) edge (m-2-3)
      (m-2-3) edge (m-2-4)
      (m-2-2) edge node[right]{$\beta_1{+}\beta_2$} (m-1-2)
      (m-2-2) edge (m-3-2)
      (m-2-3) edge (m-1-3)
      (m-2-3) edge node[right]{$[\alpha_1,\alpha_2]$} (m-3-3);
      \pullback[.6cm]{m-1-3}{m-2-3}{m-2-4}{draw,-};
      \pullback[.6cm]{m-3-2}{m-2-2}{m-2-3}{draw,-};
    \end{tikzpicture}
  \]
  where~$[\BLANK,\BLANK]$ is the ``copairing'' and all the remaining morphisms
  are either of the form~$[f,g]$ or~$f+g$.\footnote{A small lemma is necessary
  for square~$(i)$.} It is easy to split this into two simulations: one
  from~$P_1$ to~$P_3$ and the other from~$P_2$ to~$P_3$. Checking that those
  are simulations is straightforward. Conversely, we can construct a
  simulation as above from any two simulations. The constructions are
  inverse to each other (up-to isomorphism).

  This shows that~$\Plus$ is a coproduct. Because~$\PE_\Set$ is enriched
  over (large) commutative monoids coproduct is also a product.

  The proof that the polynomial~$\Zero$ is a zero object is left to the
  reader.
\end{proof}

This proof also extends to infinite coproducts:
\begin{lem}
  For a set~$K$ and polynomials~$P_k$ for~$k\in K$, the polynomial
  \[
    \bigoplus_{k\in K} P_k \quad \eqdef \quad
    \begin{tikzpicture}[baseline=(m-1-1.base)]
      \matrix (m) [matrix of math nodes, column sep=3.5em]
      { \coprod_k I_k & \coprod_k D_k & \coprod_k A_k &
    \coprod_k I_k \\ };
      \path[morphism]
      (m-1-2) edge node[above]{$\coprod_k n_k$} (m-1-1)
      (m-1-2) edge node[above]{$\coprod_k d_k$} (m-1-3)
      (m-1-3) edge node[above]{$\coprod_k a_k$} (m-1-4);
    \end{tikzpicture}
  \]
  is both the cartesian product and coproduct of the polynomials~$P_k$
  in~$\PE_\Set$.
\end{lem}

\subsection{Exponentials} 

Whenever the infinite coproduct distributes over a binary tensor~$\odot$:
\[
  \coprod_{k\geq0} (X\odot I_k)
  \quad \iso \quad
  X \odot \coprod_{k\geq0} I_k
\]
the free~$\odot$-monoid over~$I$ and free commutative~$\odot$-monoid over~$I$ are given by
\[
  \coprod_{k\geq0} I^{\odot k}
   \quad\hbox{and}\quad
  \coprod_{k\geq0} S_k(I)
\]
where~$S_k(I)$ is the coequalizer of the~$k!$ symmetries on~$I^{\odot k}$.

For the cartesian product~$\times$ on the category~$\Set$, we obtain finite
words and finite multisets, i.e. equivalence classes of words under
permutations:
\[
  I^*
  \quad\eqdef\quad
  \coprod_{k\geq0} I^k
\]
and
\[
  \Mf(I)
  \quad\eqdef\quad
  \coprod_{k\geq0} \Mf^k(I)
  \quad\eqdef\quad
  \coprod_{k\geq0} I^k/\mathfrak{S}_k
\]
where~$\mathfrak{S}_k$ is the group of permutations of~$\{1,\dots,k\}$, acting
in an obvious way on~$I^k$.
%

\medbreak
The next lemma is probably folklore among the right people, but I could find
no proof in the literature. A proof is given in appendix on
page~\pageref{app:monadSpanSet}.
\begin{lem}
 The operation~$\Mf(\BLANK)$ is the object part of a monad on~$\Span_\Set$.
  This monad gives the free commutative~$\times$-monoid in~$\Span_\Set$.
  Because~$\Span_\Set$ is self-dual, this is also the free
  commutative~$\times$-comonoid comonad.
\end{lem}
The unit and multiplication are inherited from~$\Set$:
\[\label{diag:comonadTransformation}
  w_A \eqdef
  \begin{tikzpicture}[baseline=(m-2-3.base)]
    \matrix (m) [matrix of math nodes, column sep={4em,between origins},row
    sep={4em,between origins}, text height=1.5ex, text depth=.25ex]
      {        & \One &       \\
        \One   &      & \Mf(A)\\
        };
    \path[morphism]
      (m-1-2) edge node[above right]{$\varepsilon$} (m-2-3)
      (m-1-2) edge node[above left]{$1$} (m-2-1);
  \end{tikzpicture}
  \qquad
  c_A \eqdef
  \begin{tikzpicture}[baseline=(m-2-3.base)]
    \matrix (m) [matrix of math nodes, column sep={4em,between origins},row
    sep={4em,between origins}, text height=1.5ex, text depth=.25ex]
      {                    & \Mf(A)\times\Mf(A) &        \\
        \Mf(A)\times\Mf(A) &                    & \Mf(A) \\
        };
    \path[morphism]
      (m-1-2) edge node[above right]{$\uplus$} (m-2-3)
      (m-1-2) edge node[above left]{$1$} (m-2-1);
  \end{tikzpicture}
\]
where~$\varepsilon$ picks the empty multiset and~$\uplus$ is the union of
multisets.

\bigbreak
Just as in~$\Span_\Set$, the infinite product (which is also the coproduct)
distributes over the binary tensor in~$\PE_\Set$:
\[
  \bigoplus_{k\geq0} Q\Tensor P_k
  \quad\iso\quad
  Q \Tensor \bigoplus_{k\geq0} P_k
\ .
\]
We can thus use the dual formula to get the free
commutative~$\Tensor$-comonoid.
\[
  \Bang P = \bigoplus_{k\geq0} P^k
\]
where~$P^k$ is the equalizer of all the symmetries on~$P^{\Tensor k}$.
Because the forgetful functor~$U$ is a right adjoint, it preserves products
and equalizers. Because~$U$ is monoidal, the above formula implies that it
preserves the free comonoid. Thus, the polynomial~$!P$ has domain~$\Mf(I)$
whenever~$P$ has domain~$I$.
\begin{defi}
  If~$P$ is a polynomial, define~$\Bang P$ to be the following polynomial
  \[
    \Bang P \quad\eqdef\quad
    \begin{tikzpicture}[baseline=(m-1-1.base)]
      \matrix (m) [matrix of math nodes, column sep=4em, text height=1.5ex, text depth=.25ex]
        {   \Mf(I) & D^* & A^*  & \Mf(I) \\   };
      \path[morphism]
        (m-1-2) edge node[above]{$c n^*$} (m-1-1)
        (m-1-2) edge node[above]{$d^*$} (m-1-3)
        (m-1-3) edge node[above]{$c a^*$} (m-1-4);
    \end{tikzpicture}
  \]
  where~$c$ takes a word to its orbit.
\end{defi}
Note that this is \emph{not} a pointwise application of~$\Mf$, which
would give
\[
  \begin{tikzpicture}[baseline=(m-1-1.base)]
    \matrix (m) [matrix of math nodes, column sep=4em, text height=1.5ex, text depth=.25ex]
      {   \Mf(I) & \Mf(D) & \Mf(A)  & \Mf(I) \\   };
    \path[morphism]
      (m-1-2) edge node[above]{$\Mf(n)$} (m-1-1)
      (m-1-2) edge node[above]{$\Mf(d)$} (m-1-3)
      (m-1-3) edge node[above]{$\Mf(a)$} (m-1-4);
  \end{tikzpicture}
  \ .
\]

\begin{prop}\label{prop:ExpPE}
  The free commutative~$\Tensor$-comonoid comonad on~$\Span_\Set$ lifts
  to the category~$\PE_\Set$. Its action on objects is given by~$P\mapsto!P$.
\end{prop}
Just as in~$\Span_\Set$ it is sufficient to check that
\[
  P^k\quad\eqdef\quad
  \begin{tikzpicture}[baseline=(m-1-1.base)]
    \matrix (m) [matrix of math nodes, column sep=4em, text height=1.5ex, text depth=.25ex]
      {   \Mf^k(I) & D^k & A^k  & \Mf^k(I) \\   };
    \path[morphism]
      (m-1-2) edge node[above]{$c n^k$} (m-1-1)
      (m-1-2) edge node[above]{$d^k$} (m-1-3)
      (m-1-3) edge node[above]{$c a^k$} (m-1-4);
  \end{tikzpicture}
\]
is the equalizer of the symmetries. Both the diagrammatic proof and the type
theoretic proof are possible but very tedious and we will only show
that~$P^k$ is the equalizer of the symmetries in the
category~$\PE_{\Set\sim}$, obtained from~$\PE_\Set$ by identifying any two
simulations when their spans are isomorphic. This makes the forgetful
functor~$U : \PE_{\Set\sim} \to \Span_{\Set}$ faithful, simplifying the
argument.

\medbreak
First, some notation:
\begin{itemize}
  \item tuples are denoted using the Gothic
    alphabet:~$\word u \in U^k$,~$\word i\in I^k$ etc.
  \item any permutation~$\sigma\in\mathfrak{S}_k$ induces a natural
    transformation~$\BLANK^k \to \BLANK^k$,
  \item $c:\BLANK^* \to \Mf(\BLANK)$ is the natural transformation sending
    a tuple to its orbit,
  \item for any set~$X$,~$s : \Mf(X)\rightarrowtail X^*$ is a section of~$c_X
    : X^* \twoheadrightarrow \Mf(X)$;\footnote{this transformation cannot be
    made natural}
\end{itemize}
and a preliminary lemma:
\begin{lem}\label{lem:permutationsSetPullback}
  In~$\Set$, suppose~$h:U^k \to U^k$ sends any element of~$U^k$ to a
  permutation of itself, i.e.,
  \[
    \begin{tikzpicture}[baseline=(m-1-1)]
      \matrix (m) [matrix of math nodes, column sep={5em,between origins}, row
      sep={3em,between origins},text height=1.5ex, text depth=.25ex]
      { U^k & \Mf^k(U) \\
        U^k &          \\};
      \path[morphism]
        (m-1-1) edge node[above]{$c$} (m-1-2)
        (m-2-1) edge node[left]{$h$} (m-1-1)
        (m-2-1) edge node[below right]{$c$} (m-1-2);
    \end{tikzpicture}
    \ .
  \]
  For any~$g : V \to U$, we can find a~$\rho : V^k \to V^k$ with the same
  property, i.e., with~$c \rho = c$ such that:
  \[
    \begin{tikzpicture}[baseline=(m-2-2)]
      \matrix (m) [matrix of math nodes, column sep={4.5em,between origins}, row
      sep={4.5em,between origins},text height=1.5ex, text depth=.25ex]
      { \Mf^k(V) & V^k & U^k \\
                 & V^k & U^k \\};
      \path[morphism]
        (m-1-2) edge node[above]{$f^k$} (m-1-3)
        (m-2-2) edge node[below]{$f^k$} (m-2-3)
        (m-1-2) edge node[left]{$\rho$} (m-2-2)
        (m-1-3) edge node[right]{$h$} (m-2-3)
        (m-1-2) edge node[above]{$c$} (m-1-1)
        (m-2-2) edge node[below left]{$c$} (m-1-1);
        \pullback{m-1-3}{m-1-2}{m-2-2}{draw,-};
    \end{tikzpicture}
    \ .
  \]
\end{lem}
\begin{proof}
  Define~$\rho : \word v \mapsto \sigma_{f^k(\word v)}(\word v)$,
  where~$\sigma_{\word u}$ is any permutation s.t.~$f^k(\word u) =
  \sigma_{\word u}(\word u)$.
  This~$\rho$ makes the diagram commute. To show that the square is a
  pullback, we construct mediating arrows as follows: given
  \[
    \begin{tikzpicture}[baseline=(m-3-3.base)]
      \matrix (m) [matrix of math nodes, column sep={3.5em,between origins}, row
      sep={3.5em,between origins},text height=1.5ex, text depth=.25ex]
      {X &     &     \\
         & V^k & U^k \\
         & V^k & U^k \\};
      \path[morphism]
      (m-2-2) edge node[above]{$f^k$} (m-2-3)
      (m-2-2) edge node[left]{$\rho$} (m-3-2)
      (m-3-2) edge node[below]{$f^k$} (m-3-3)
      (m-2-3) edge node[right]{$h$} (m-3-3);
      \path[morphism,bend left=25]
      (m-1-1) edge node[above]{$g_1$} (m-2-3);
      \path[morphism,bend left=-35]
      (m-1-1) edge node[left]{$g_2$} (m-3-2) ;
      \path[morphism,dashed]
      (m-1-1) edge node[above right]{$\gamma$} (m-2-2);
      \pullback[.5cm]{m-3-2}{m-2-2}{m-2-3}{draw,-}
    \end{tikzpicture}
  \]
  we put~$\gamma(x) \eqdef \sigma^{-1}_{g_1(x)}\big(g_2(x)\big)$.
  We have
  \[
    f^k \gamma (x)
    =
    f^k\sigma^{-1}_{g_1(x)} g_2(x)
    =
    \sigma^{-1}_{g_1(x)} f^k g_2(x)
    =
    \sigma^{-1}_{g_1(x)} g_1 h(x)
    =
    g_1(x)
  \]
  where the last equality comes from~$h g_1(x)=\sigma_{g_1(x)}\big(g_1(x)\big)$.
  For the second triangle:
  \[
    \rho \gamma (x)
    =
    \sigma_{f^k\gamma(x)} \gamma(x)
    =
    \sigma_{g_1(x)} \gamma(x)
    =
    \sigma_{g_1(x)} \sigma^{-1}_{g_1(x)}g_2(x)
    =
    g_2(x)\ .
  \]
  Moreover, for any other mediating~$\gamma'$, we must have
  \[
    g_2(x)
    \quad=\quad
    \rho \gamma'(x)
    \quad=\quad
    \sigma_{f^k \gamma'(x)}\gamma'(x)
    \quad=\quad
    \sigma_{g_1(x)}\gamma'(x)
  \]
  which implies that~$\gamma'=\gamma$.
\end{proof}

\begin{proof}[Proof of Proposition~\ref{prop:ExpPE}]
  We need to show that~$P^k$ is the equalizer of the symmetries on~$P^{\Tensor k}$. We
  first need to make~$\hat c$ into a simulation from~$P^k$ to~$P^{\Tensor
  k}$, i.e., we need to define~$\alpha$,~$\beta$ and~$\gamma$ filling the
  diagram
  \[
    \begin{tikzpicture}[baseline=(m-3-2)]
      \matrix (m) [matrix of math nodes, column sep={4em,between origins}, row
      sep={4em,between origins},text height=1.5ex, text depth=.25ex]
      { P^k:          &\Mf^k(I)& D^k & A^k &\Mf^k(I)  \\
                      & I^k    & Y   & X   & I^k      \\
       P^{\Tensor k}: & I^k    & D^k & A^k & I^k      \\};
      \path[morphism]
      (m-1-3) edge (m-1-2)
      (m-1-3) edge (m-1-4)
      (m-1-4) edge (m-1-5)
      (m-3-3) edge node[below]{$n^k$} (m-3-2)
      (m-3-3) edge node[below]{$d^k$} (m-3-4)
      (m-3-4) edge node[below]{$a^k$} (m-3-5);
      \path[morphism]
      (m-2-2) edge node[left]{$c$} (m-1-2)
      (m-2-2) edge node[left]{$1$} (m-3-2)
      (m-2-5) edge node[right]{$c$} (m-1-5)
      (m-2-5) edge node[right]{$1$} (m-3-5)
      (m-2-3) edge node[over]{$\gamma$} (m-2-2)
      (m-2-3) edge (m-2-4)
      (m-2-4) edge (m-2-5)
      (m-2-3) edge node[over]{$\beta$} (m-1-3)
      (m-2-3) edge (m-3-3)
      (m-2-4) edge (m-1-4)
      (m-2-4) edge node[over]{$\alpha$} (m-3-4);
      \pullback[.6cm]{m-1-4}{m-2-4}{m-2-5}{draw,-};
      \pullback[.6cm]{m-3-3}{m-2-3}{m-2-4}{draw,-};
    \end{tikzpicture}
    \ .
  \]
  The set~$X$ is (isomorphic to)~$\big\{(\word i,\word a)\ \big|\ \word i \sim
  a^k(\word a)\big\}$ and the function~$\alpha$ sends~$(\word i,\word a)$
  to~$\sigma_{\word i,\word a}(\word a)$ where~$\sigma_{\word i,\word a}$ is a
  permutation such that~$\sigma_{\word i,\word a}\big(a^k(\word a)\big) = \word
  i$; and the set~$Y$ is (isomorphic to)~$\big\{(\word i,\word a,\word d)\
  \big|\ \word i \sim a^k(\word a), d^k(\word d)=\word a\big\}$. The
  function~$\beta$ sends~$(\word i,\word a,\word d)$ to~$\sigma^{-1}_{\word
  i,\word a}(\word d)$ and the function~$\gamma$ sends~$(\word i,\word a,\word
  d)$ to~$n^k(\word d)$.

  \medbreak
  Now, given a simulation from~$Q$ to~$P^{\Tensor k}$
  \begin{equation}\label{eqn:SimEqualizes}
    \begin{tikzpicture}[baseline=(m-2-2)]
      \matrix (m) [matrix of math nodes, column sep={3.5em,between origins}, row
      sep={3.5em,between origins},text height=1.5ex, text depth=.25ex]
      { Q:            & J  &  E  &  B  &  J  \\
                     & R &\cdotp&\cdotp& R \\
       P^{\Tensor k}: & I^k & D^k & A^k & I^k\\};
      \path[morphism]
      (m-1-3) edge (m-1-2)
      (m-1-3) edge (m-1-4)
      (m-1-4) edge (m-1-5)
      (m-3-3) edge node[below]{$n^k$} (m-3-2)
      (m-3-3) edge node[below]{$d^k$} (m-3-4)
      (m-3-4) edge node[below]{$a^k$} (m-3-5);
      \path[morphism]
      (m-2-2) edge node[left]{$j$} (m-1-2)
      (m-2-2) edge node[left]{$f$} (m-3-2)
      (m-2-5) edge node[right]{$j$} (m-1-5)
      (m-2-5) edge node[right]{$f$} (m-3-5)
      (m-2-3) edge node[over]{$\gamma$} (m-2-2)
      (m-2-3) edge (m-2-4)
      (m-2-4) edge (m-2-5)
      (m-2-3) edge node[over]{$\beta$} (m-1-3)
      (m-2-3) edge node[over]{$l$} (m-3-3)
      (m-2-4) edge (m-1-4)
      (m-2-4) edge node[over]{$\alpha$} (m-3-4);
      \pullback[.6cm]{m-1-4}{m-2-4}{m-2-5}{draw,-};
      \pullback[.6cm]{m-3-3}{m-2-3}{m-2-4}{draw,-};
    \end{tikzpicture}
  \end{equation}
  which equalizes the symmetries, we need to construct a simulation from~$Q$
  to~$P^k$. That~(\ref{eqn:SimEqualizes}) equalizes the symmetries implies in
  particular that
  \begin{equation}\label{eqn:spanEqualizesEquiv}
    \forall \sigma\in\mathfrak{S}_k\ \exists H\quad
    \begin{tikzpicture}[baseline=(m-2-1.base)]
      \matrix (m) [matrix of math nodes, column sep={5em,between origins},row
      sep={3em,between origins}, text height=1.5ex, text depth=.25ex]
      {    & R &   \\
       I^k &   & J \\
           & R &   \\};
      \path[morphism]
        (m-1-2) edge node[above left]{$f$} (m-2-1)
        (m-1-2) edge node[above right]{$j$} (m-2-3)
        (m-3-2) edge node[below left]{$\sigma f$} (m-2-1)
        (m-3-2) edge node[below right]{$j$} (m-2-3);
      \path[morphism]
        (m-3-2) edge node[sloped,over]{$\sim$} node[above,sloped]{$H$} (m-1-2);
    \end{tikzpicture}
    \ .
  \end{equation}

  The simulation from~$Q$ to~$P^k = \Mf^k(I)\leftarrow D^k\to A^k\to\Mf^k(I)$ is
  constructed in several steps as indicated by the small numbers in parenthesis:

  \newcommand\mini[1]{\ensuremath{{\scriptscriptstyle(#1)}}}
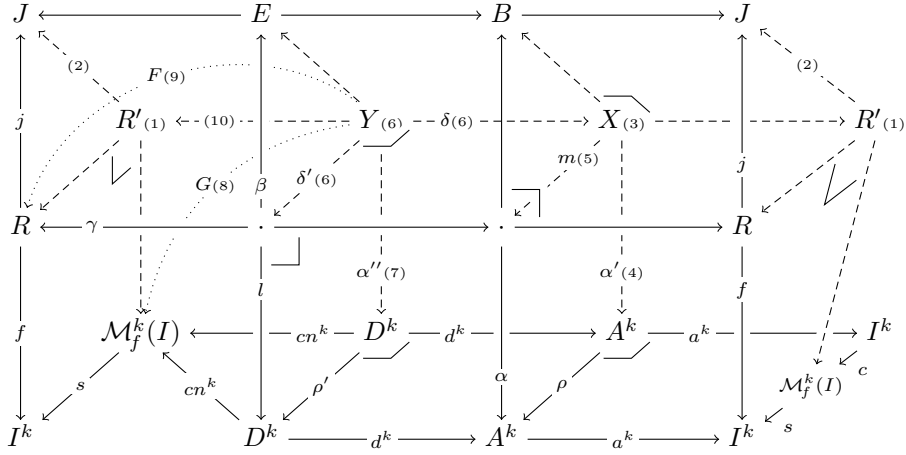
  \[
    \begin{tikzpicture}[baseline=(m-1-1)]
      \matrix (m) [matrix of math nodes, column sep={4.5em,between origins}, row
      sep={4em,between origins},text height=1.5ex, text depth=.25ex]
      { J  &&  E  &&  B  &&  J  \\
        & R'\mini{1} &&Y\mini6&&X\mini3&& R'\mini{1}     \\
        R  &&\cdot&&\cdot&&  R  \\
        &\Mf^k(I)&&D^k&&A^k&&I^k     \\
       I^k && D^k && A^k && I^k\\};

      \node (A) at ($(m-5-7)!.5!(m-4-8)$) {$\scriptstyle\Mf^k(I)$};
      \path[morphism]
        (m-4-2) edge node[over]{$s$} (m-5-1)
        (m-4-8) edge node[below right]{$c$} (A)
        (A) edge node[below right]{$s$} (m-5-7) ;

      \path[morphism]
      (m-4-6) edge node[over]{$\rho$} (m-5-5)
      (m-4-4) edge node[over]{$\rho'$} (m-5-3)
      (m-5-3) edge node[over]{$c n^k$} (m-4-2) ;
      \pullback[.5cm]{m-4-6}{m-4-4}{m-5-3}{draw,-};
      \pullback[.5cm]{m-4-8}{m-4-6}{m-5-5}{draw,-};

      \path[morphism]
        (m-4-4) edge node[near start,over]{$c n^k$} (m-4-2)
        (m-4-4) edge node[near start,over]{$d^k$} (m-4-6)
        (m-4-6) edge node[near start,over]{$a^k$} (m-4-8) ;

      \path[morphism,densely dashed]
        (m-2-2) edge node[over]{\mini2} (m-1-1)
        (m-2-2) edge (m-4-2)
        (m-2-2) edge (m-3-1);
      \pullback[.5cm]{m-3-1}{m-2-2}{m-4-2}{draw,-};

      \path[morphism,densely dashed]
        (m-2-8) edge node[over]{\mini2} (m-1-7)
        (m-2-8) edge[crossover] (A)
        (m-2-8) edge (m-3-7);
      \pullback{m-3-7}{m-2-8}{A}{draw,-};

      \path[morphism,densely dashed]
        (m-2-6) edge (m-2-8)
        (m-2-6) edge (m-1-5);
      \pullback[.5cm]{m-2-8}{m-2-6}{m-1-5}{draw,-};

      \path[morphism,densely dashed]
        (m-2-6) edge node[near start,over]{$m \mini5$} (m-3-5) ;

      \path[morphism,dotted]
        (m-2-4) edge[bend right=55] node[over]{$F \mini9$} (m-3-1)
        (m-2-4) edge[bend right=35] node[over]{$G \mini8$} (m-4-2);
      \path[morphism,densely dashed]
        (m-2-4) edge node[near end,over]{\mini{10}} (m-2-2);

      \path[morphism,densely dashed]
        (m-2-6) edge node[near end,over]{$\alpha' \mini4$} (m-4-6)
        (m-2-4) edge node[near end,over]{$\alpha'' \mini7$} (m-4-4) ;

      \path[morphism, densely dashed]
        (m-2-4) edge (m-1-3)
        (m-2-4) edge node[near start,over]{$\delta \mini6$} (m-2-6)
        (m-2-4) edge node[over]{$\delta' \mini6$} (m-3-3)
        ;
      \pullback[.5cm]{m-2-6}{m-2-4}{m-3-3}{draw,-,crossover};

      \path[morphism]
      (m-1-3) edge (m-1-1)
      (m-1-3) edge (m-1-5)
      (m-1-5) edge (m-1-7)

      (m-5-3) edge node[over]{$d^k$} (m-5-5)
      (m-5-5) edge node[over]{$a^k$} (m-5-7)

      (m-3-3) edge[crossover] node[near end,over]{$\gamma$} (m-3-1)
      (m-3-3) edge[crossover] (m-3-5)
      (m-3-5) edge[crossover] (m-3-7)

      (m-3-1) edge node[over]{$j$} (m-1-1)
      (m-3-1) edge node[over]{$f$} (m-5-1)
      (m-3-3) edge[crossover] node[very near start,over]{$\beta$} (m-1-3)
      (m-3-3) edge[crossover] node[near start,over]{$l$} (m-5-3)
      (m-3-5) edge (m-1-5)
      (m-3-5) edge[crossover] node[near end,over]{$\alpha$} (m-5-5)
      (m-3-7) edge[crossover] node[near start,over]{$j$} (m-1-7)
      (m-3-7) edge[crossover] node[near start,over]{$f$} (m-5-7);
      \pullback[.5cm]{m-3-7}{m-3-5}{m-1-5}{draw,-,crossover};
      \pullback[.5cm]{m-3-5}{m-3-3}{m-5-3}{draw,-};

    \end{tikzpicture}
  \]
  where the bottom layer comes from~Lemma~\ref{lem:permutationsSetPullback}
  and:
  \begin{itemize}

    \item[\mini1] $R'$ is constructed by pullback;
    \item[\mini2] is obtained by composition;
    \item[\mini3] $X$ is obtained by pullback;
    \item[\mini4] $\alpha'$ is the mediating arrow, where~$\mini{3} \to I^k$ is~$\mini{3} \to R' \to \Mf^k(I) \to I^k$;
    \item[\mini5] $m$ is the mediating arrow;
    \item[\mini6] $Y$, $\delta$ and~$\delta'$ are constructed by pullback;
    \item[\mini7] $\alpha''$ is the mediating arrow;
    \item[\mini8] $G$ sends~$y$ to~$c n^k \alpha''(y)$;
    \item[\mini9] $F$ sends~$y$ to~$H^{-1}_{f \gamma
      \delta'(y)}\big(\gamma \delta' (y)\big)$, where for~$\word i\in
      I^k$,~$H_{\word i}$ is the automorphism in
      diagram~$(\ref{eqn:spanEqualizesEquiv})$ corresponding
      to~$\sigma_{\word i}$, chosen such that~$\sigma_{\word i}=s c(\word i)$;
    \item[\mini{10}] is the mediating arrow corresponding to~$F$ and~$G$.
  \end{itemize}
  We only need to check that~\mini8 and~\mini9 make the appropriate diagram
  commute, i.e., that~$s G = f F$:
  \begin{eqnarray*}
    s G(y)
    &\quad=\quad&
    s c n^k \alpha''(y)\\
    &\quad=\quad&
    s c n^k \rho' \alpha''(y)\\
    &\quad=\quad&
    s c n^k l \delta'(y)\\
    (\star)
    &\quad=\quad&
    s c f \gamma \delta'(y)\\
    &\quad=\quad&
    \sigma_{f \gamma \delta'(y)} f \gamma \delta'(y)\\
    &\quad=\quad&
    f H^{-1}_{f \gamma \delta'(y)}\gamma\delta'(y)\\
    &\quad=\quad&
    f F(y)
  \end{eqnarray*}
  where equality~$(\star)$ comes from diagram~(\ref{eqn:SimEqualizes}). Because
  square~$(\delta,\alpha',\alpha'',d^k)$ is a pullback, this gives a simulation
  from~$Q$ to~$P^k$.

  \medbreak
  That this simulation is the mediating arrow for the equalizer diagram
  in~$\PE_{\Set\sim}$ follows
  from the fact that its corresponding span is indeed the mediating span
  in~$\Span_\Set$, together with the fact that the forgetful
  functor~$U:\PE_{\Set\sim} \to \Span_\Set$ is faithful.

\end{proof}



\section*{Concluding Remarks} 

\subsection*{Toward Differential Logic} 

As noted in section~\ref{sub:enrichedStructure}, the category~$\PE_\Set$ is
enriched over large commutative monoids. Moreover,~$\PE_\Set$ has enough
duality to make~$\Bang P$ into a commutative~$\Tensor$-\emph{monoid}. These are key
features when one interprets differential logic \cite{difflamb,diffnet}.
However, trying to lift the differential structure of~$\Rel$, the category of
sets and relations to the category~$\Span_\Set$ fails as the candidate for the
deriving transformation~\cite{diffCategories}:
\[
  \dd_X \quad : \quad
  X \Tensor \Bang X \to \Bang X
    \dd_X \quad\eqdef\quad
    \begin{tikzpicture}[baseline=(m-2-3)]
      \matrix (m) [matrix of math nodes, column sep={4em,between origins}, row
      sep={4em,between origins},text height=1.5ex, text depth=.25ex]
        {                & X\times \Mf(X) & \\
          X\times \Mf(X) &                & \Mf(X)\\};
      \path[morphism]
        (m-1-2) edge node[over]{$1$} (m-2-1)
        (m-1-2) edge node[over]{$\at$} (m-2-3);
    \end{tikzpicture}
\]
is only \emph{lax natural}. It is the only obstruction to get a differential
category in the sense of Blute, Cockett and Seely \cite{diffCategories} as
the four coherence conditions seem to hold both in~$\Span_\Set$ and
in~$\PE_\Set$, even if the full proof for the later is rather long. (As with
Proposition~\ref{prop:ExpPE}, the proof is much simpler for the
category~$\PE_{\Set\sim}$.)
Whether lax naturality is enough to model differential logic is still to
be investigated.


\subsection*{Extensional Version: Polynomial Functors and Simulation Cells}

This work was very ``intensional'' in that it only dealt with polynomial
\emph{diagrams} and not at all with polynomial \emph{functors}. The different
notions presented here have a more ``extensional'' version which does not rely
on knowing a particular representation of the polynomial functors. In
particular, the tensor and the linear arrow of two polynomial functors can be
defined by universal properties. The corresponding category~$\PEFun_\C$, has
polynomial functors as objects, and morphisms (simulations) are given by cells
of the form
\[
    \begin{tikzpicture}[baseline=(m-2-1.base)]
      \matrix (m) [matrix of math nodes, column sep=3em, row sep=3em]
        { \Sl\C{I_1} & \Sl\C{J_1}\\
          \Sl\C{I_2} & \Sl\C{J_2}\\};
      \path[morphism]
        (m-1-1) edge node[above]{$P_1$} (m-1-2)
        (m-2-1) edge node[below]{$P_2$} (m-2-2);
      \path[morphism]
        (m-1-1) edge node[over]{$L$} (m-2-1)
        (m-1-2) edge node[over]{$L$} (m-2-2);
      \draw[twocell]
        (m-1-2) -- (m-2-1);
      \draw ($(m-1-2)!.5!(m-2-1)$) node[over]{$\scriptstyle\alpha$};
    \end{tikzpicture}
\]
where~$L$ is a ``linear'' polynomial functor. Composition is simply
obtained by pasting such cells vertically. See the
upcoming~\cite{polyFunctors} for details.


\bibliographystyle{amsalpha}
\bibliography{poly}


\newpage
\appendix

\section{Free Commutative $\times$-Monoid in $\Span_\Set$} 
\label{app:monadSpanSet}

\begin{lem}
 The operation~$\Mf(\BLANK)$ is the object part of a monad on~$\Span_\Set$.
  This monad gives the free commutative~$\times$-monoid in~$\Span_\Set$.
  Because~$\Span_\Set$ is self-dual, this is also the free
  commutative~$\times$-comonoid comonad.
\end{lem}

\begin{proof}
  Because the coproduct distributes over~$\times$, and because coproducts
  in~$\Span_\Set$ are computed as in~$\Set$, we only need to show
  that~$\Mf^k(I)$ is the coequalizer of all the symmetries on~$I^k$.
  Let~$c:I^* \twoheadrightarrow \Mf(I)$ be the function sending a word to its corresponding
  multiset, and let~$s:\Mf(I)\rightarrowtail I^*$ be a section of~$c$, i.e., a function
  choosing a representative for each equivalence class. This gives rise to a
  pair retraction/section in~$\Span_\Set$:
  \[
    \begin{tikzpicture}[baseline=(m-1-1.base)]
      \matrix (m) [matrix of math nodes, column sep={8em,between origins},row
      sep={3em,between origins}, text height=1.5ex, text depth=.25ex]
      {  I^k & \Mf^k(I) \\};
      \path[morphism]
        (m-1-1) edge node{$\scriptstyle|$} node[above]{$\hat c$} (m-1-2);
    \end{tikzpicture}
    \quad\hbox{is}\quad
    \begin{tikzpicture}[baseline=(m-2-1.base)]
      \matrix (m) [matrix of math nodes, column sep={4em,between origins},row
      sep={3em,between origins}, text height=1.5ex, text depth=.25ex]
      {     & I^k &        \\
        I^k &     & \Mf^k(I) \\};
      \path[morphism]
        (m-1-2) edge node[above left]{$1$} (m-2-1)
        (m-1-2) edge node[above right]{$c$} (m-2-3);
    \end{tikzpicture}
  \]
  and
  \[
    \begin{tikzpicture}[baseline=(m-1-1.base)]
      \matrix (m) [matrix of math nodes, column sep={8em,between origins},row
      sep={3em,between origins}, text height=1.5ex, text depth=.25ex]
      {  \Mf^k(I) & I^k \\};
      \path[morphism]
        (m-1-1) edge node{$\scriptstyle|$} node[above]{$\hat s$} (m-1-2);
    \end{tikzpicture}
    \quad\hbox{is}\quad
    \begin{tikzpicture}[baseline=(m-2-1.base)]
      \matrix (m) [matrix of math nodes, column sep={4em,between origins},row
      sep={3em,between origins}, text height=1.5ex, text depth=.25ex]
      {          & \Mf^k(I) &     \\
        \Mf^k(I) &          & I^k \\};
      \path[morphism]
        (m-1-2) edge node[above left]{$1$} (m-2-1)
        (m-1-2) edge node[above right]{$s$} (m-2-3);
    \end{tikzpicture}
    \ .
  \]

  We'll show that~$\hat c$ is the coequalizer of the symmetries: consider
  \begin{equation}\label{eqn:spanCoequalizerDiagram}
    \begin{tikzpicture}[baseline=(m-1-5.base)]
      \matrix (m) [matrix of math nodes, column sep={6em,between origins},row
      sep={4em,between origins}, text height=1.5ex, text depth=.25ex]
      {  I^k & \lower2.5pt\vdots\quad\hbox{\tiny
         $k\exclamation$ symmetries}\quad\lower2.5pt\vdots & I^k && \Mf^k(I)\\
         &&&& J \\};
      \path[morphism]
        (m-1-3) edge node[sloped]{$\scriptstyle|$} node[below left]{$\phi$} (m-2-5)
        (m-1-3) edge node{$\scriptstyle|$} node[above]{$\hat c$} (m-1-5);
      \path[morphism]
        (m-1-1.north east) edge node{$\scriptstyle|$} node[above]{$\sigma$} (m-1-3.north west)
        (m-1-1.south east) edge node{$\scriptstyle|$} node[below]{$\sigma'$} (m-1-3.south west);
      \path[morphism,densely dashed]
        (m-1-5) edge node[right]{$?\psi$} (m-2-5);
    \end{tikzpicture}
  \end{equation}
  where the~$\sigma$s are spans with the identity for right leg and permutations
  for left leg. It is immediate that~$\hat c$ coequalizes them. Suppose
  moreover that~$\phi=I^k \leftarrow R \to J$ coequalizes them, i.e.,
  \begin{equation}\label{eqn:coequalizesWitness}
    \forall \sigma\in\mathfrak{S}_k\ \exists H\quad
    \begin{tikzpicture}[baseline=(m-2-1.base)]
      \matrix (m) [matrix of math nodes, column sep={5em,between origins},row
      sep={3em,between origins}, text height=1.5ex, text depth=.25ex]
      {    & R &   \\
       I^k &   & J \\
           & R &   \\};
      \path[morphism]
        (m-1-2) edge node[above left]{$f$} (m-2-1)
        (m-1-2) edge node[above right]{$j$} (m-2-3)
        (m-3-2) edge node[below left]{$\sigma f$} (m-2-1)
        (m-3-2) edge node[below right]{$j$} (m-2-3);
      \path[morphism]
        (m-3-2) edge node[sloped,over]{$\sim$} node[above,sloped]{$H$} (m-1-2);
    \end{tikzpicture}
    \ .
  \end{equation}
  To close the triangle in~(\ref{eqn:spanCoequalizerDiagram}), put~$\psi
  \eqdef \phi \hat s$. We need to check that~$\psi \hat c
  = \phi$, i.e., that~$\phi\hat s\hat c = \phi$. We have
  \[
    \phi\hat s\hat c
    \qquad=\qquad
    \begin{tikzpicture}[baseline=(m-2-2.base)]
      \matrix (m) [matrix of math nodes, column sep={3em,between origins},row
      sep={3em,between origins}, text height=1.5ex, text depth=.25ex]
      {    & & R' & &         \\
           & I^k &     & R &   \\
       I^k &     & I^k &   & J \\};
      \path[morphism]
        (m-1-3) edge node[above left]{$\pi_1$} (m-2-2)
        (m-1-3) edge node[above right]{$\pi_2$} (m-2-4)
        (m-2-2) edge node[over]{$1$} (m-3-1)
        (m-2-2) edge node[over]{$s c$} (m-3-3)
        (m-2-4) edge node[over]{$f$} (m-3-3)
        (m-2-4) edge node[over]{$j$} (m-3-5);
      \pullback[0.6cm]{m-2-2}{m-1-3}{m-2-4}{draw,-}
    \end{tikzpicture}
  \]
  where
  \[
    R'\quad\eqdef\quad\left\{\Big. (\word i,r) \in I^k\times R \ \middle|\ s c (\word i) = f(r)\right\}
    \ .
  \]
  To show that this span is equal to~$\phi$, we need to find an isomorphism
  between~$R$ and~$R'$ s.t.
  \begin{equation}\label{eqn:coequalizerTriangles}
    \begin{tikzpicture}[baseline=(m-2-1.base)]
      \matrix (m) [matrix of math nodes, column sep={5em,between origins},row
      sep={3em,between origins}, text height=1.5ex, text depth=.25ex]
      {    & R' &   \\
       I^k &   & J \\
           & R &   \\};
      \path[morphism]
        (m-1-2) edge node[above left]{$\pi_1$} (m-2-1)
        (m-1-2) edge node[above right]{$j\pi_2$} (m-2-3)
        (m-3-2) edge node[below left]{$f$} (m-2-1)
        (m-3-2) edge node[below right]{$j$} (m-2-3);
      \path[morphism]
      (m-3-2) edge node[sloped,over]{$\sim$} node[above,sloped]{$\varepsilon$} (m-1-2);
    \end{tikzpicture}
    \ .
  \end{equation}
  To do that, note that for any word~$\word i\in I^k$, the word~$s c(\word i)$ is a
  permutation of~$\word i$. For any such~$\word i$, choose
  some~$\sigma_{\word i}\in\mathfrak{S}_k$ s.t.~$s c(\word i) = \sigma_{\word i}(\word i)$, and
  define~$\varepsilon$ to be the function
  \[
    r\ \mapsto\ \Big(f(r)\ ,\ H^{-1}_{f(r)}(r)\Big)
  \]
  where~$H_{\word i}$ is the automorphism on~$R$ corresponding to the
  permutation~$\sigma_{\word i}^{-1}$ in~(\ref{eqn:coequalizesWitness}).
  That~$\pi_1 \varepsilon=f$ is trivial, and that~$j \pi_2 \varepsilon = j$
  follows from diagram~\ref{eqn:coequalizesWitness}.
  The inverse of~$\varepsilon$ is the function~ $(\word i,r) \mapsto H_{\word
  i}(r)$:
  \[
    r
    \quad\mapsto\quad
    \Big(f(r),H^{-1}_{f(r)}(r)\Big)
    \quad\mapsto\quad
    H_{f(r)}H^{-1}_{f(r)}(r)
    =
    r
  \]
  and
  \[
    (\word i, r)
    \quad\mapsto\quad
    H_{\word i}(r)
    \quad\mapsto\quad
    \Big( f H_{\word i}(r) , H^{-1}_{f H_{\word i}(r)}H_{\word i}(r)\Big)
    =(\word i,r)
  \]
  where the equality follows from
  \[
    f H_{\word i}(r)
    \quad=\quad
    \sigma^{-1}_{\word i} f (r)
    \quad=\quad
    \sigma^{-1}_{\word i} s c (\word i)
    \quad=\quad
    \sigma^{-1}_{\word i} \sigma_{\word i}(\word i)
    \quad=\quad
    \word i
    \ .
  \]

  \bigbreak
  Because~$\hat s$ is a section of~$\hat c$, this~$\psi$ is unique:
  if~$\psi'\hat c = \phi$, we have~$\psi'=\psi'\hat c\hat s = \phi\hat s$.
  This concludes the proof.

\end{proof}

\pdfinfo{
  /Title (A Differential Category of Polynomial Diagram)
  /Author (Pierre Hyvernat)
  /Subject (Category Theory)
  /Keywords (category theory; polynomial functors; polynomial diagrams; linear logic)}

\end{document}